\documentclass[11pt]{article}
\usepackage{amsmath,amssymb,amsthm}
\usepackage[mathic]{mathtools}
\usepackage[T1]{fontenc}
\usepackage{lmodern}
\usepackage{xcolor}
\usepackage{dsfont}
\usepackage{microtype}
\usepackage[margin=1in]{geometry}
\usepackage[backgroundcolor=gray!20, linecolor=blue, textsize=scriptsize]{todonotes}
\usepackage[linesnumbered,ruled,vlined,boxed]{algorithm2e}
\newlength{\commentWidth}
\newcommand{\atcp}[1]{\tcp*[r]{\makebox[\commentWidth]{\textit{#1}\hfill}}}
\let\oldnl\nl
\newcommand{\nonl}{\renewcommand{\nl}{\let\nl\oldnl}}
\DontPrintSemicolon
\SetKwInOut{Input}{Input}\SetKwInOut{Output}{Output}

\usepackage{paralist}

\newcommand{\ignore}[1]{}

\makeatletter
\renewcommand\paragraph{\@startsection{subparagraph}{5}{\z@}%
                                      {1.5ex \@plus1ex \@minus .2ex}%
                                      {-1em}%
                                      {\normalsize\bfseries}}
\makeatother
\makeatletter
\renewcommand\subparagraph{%
 \@startsection {subparagraph}{5}{\z@ }{3.25ex \@plus 1ex
 \@minus .2ex}{-1em}{\normalfont \normalsize \bfseries }}%
\makeatother
\usepackage[ocgcolorlinks]{hyperref} % Option ocgcolorlinks makes links only colored on screen and not on printouts
\colorlet{DarkRed}{red!50!black}
\colorlet{DarkGreen}{green!50!black}
\colorlet{DarkBlue}{blue!50!black}

\hypersetup{
	linkcolor = DarkRed,
	citecolor = DarkGreen,
	urlcolor = DarkBlue,
	bookmarksnumbered = true,
	linktocpage = true
}

\newtheorem{theorem}{Theorem}[section]
\newtheorem{lemma}[theorem]{Lemma}
\newtheorem{definition}[theorem]{Definition}
\newtheorem{corollary}[theorem]{Corollary}
\newtheorem{observation}[theorem]{Observation}

\DeclareMathOperator{\diam}{diam}
\DeclareMathOperator{\hop}{hop}

\DeclareMathOperator{\poly}{poly}
\DeclareMathOperator{\polylog}{polylog}
\DeclareMathOperator{\dist}{d}

\DeclareMathOperator{\E}{E}
\DeclareMathOperator{\Exp}{Exp}
\DeclareMathOperator{\str}{str}
\DeclareMathOperator{\cut}{cut}

\DeclareMathOperator{\Ex}{\mathbb{E}}

\DeclareMathOperator{\Load}{load}

\newcommand{\Seq}[1]{\mathbf{#1}}

\newcommand{\Ecut}{E^{\cut}}

\let\epsilon\varepsilon
\let\eps\varepsilon

\newcommand{\ones}{\mathds{1}}

\newcommand{\NN}{\ensuremath{\mathbb{N}}}

\newcommand{\ZZ}{\ensuremath{\mathbb{Z}}}

\newcommand{\UUU}{\mathcal{U}}
\newcommand{\VVV}{\mathcal{V}}
\newcommand{\TTT}{\mathcal{T}}

\usepackage{xspace}
\newcommand{\Congest}{$\mathsf{CONGEST}$\xspace}
\newcommand{\Local}{$\mathsf{LOCAL}$\xspace}
\newcommand{\Pram}{$\mathsf{PRAM}$\xspace}
\newcommand{\decomp}{\texttt{ts\_decompose}\xspace}

\usepackage{tikz}
\usetikzlibrary{decorations.pathreplacing}
\usetikzlibrary{plotmarks}
\usetikzlibrary{positioning,automata,arrows}
\usetikzlibrary{shapes.geometric}
\usetikzlibrary{decorations.markings}
\usetikzlibrary{positioning, shapes, snakes, arrows}

\PassOptionsToPackage{usenames,dvipsnames,svgnames}{xcolor}
\definecolor{orange}{RGB}{235,90,0}
\definecolor{darkorange}{RGB}{175,30,0}
\definecolor{turkis}{RGB}{131,182,182}
\definecolor{darkturkis}{RGB}{31,82,82}
\definecolor{green}{RGB}{102,180,0}
\definecolor{darkgreen}{RGB}{51,90,0}
\definecolor{myblue}{RGB}{0,0,213}
\definecolor{mydarkblue}{RGB}{0,0,100}
\definecolor{mybrightblue}{HTML}{74B0E4}
\definecolor{mybrighterblue}{HTML}{B3EAFA}
\definecolor{lila}{RGB}{102,0,102}
\definecolor{darkred}{RGB}{139,0,0}
\definecolor{darkyellow}{RGB}{188,135,2}
\definecolor{brightgray}{RGB}{200,200,200}
\definecolor{darkgray}{RGB}{50,50,50}
\definecolor{amaranth}{rgb}{0.9, 0.17, 0.31}
\definecolor{alizarin}{rgb}{0.82, 0.1, 0.26}
\definecolor{amber}{rgb}{1.0, 0.75, 0.0}
\definecolor{green(ryb)}{rgb}{0.4, 0.69, 0.2}
\definecolor{hanblue}{rgb}{0.27, 0.42, 0.81}
\definecolor{grannysmithapple}{rgb}{0.66, 0.89, 0.63}

\newcommand\irregularcircle[2]{% radius, irregularity
  \pgfextra {\pgfmathsetmacro\len{(#1)+rand*(#2)}}
  +(0:\len pt)
  \foreach \a in {10,20,...,350}{
    \pgfextra {\pgfmathsetmacro\len{(#1)+rand*(#2)}}
    -- +(\a:\len pt)
  } -- cycle
}

\title{Low Diameter Graph Decompositions\\
by Approximate Distance Computation}

\author{
	Ruben Becker\thanks{Gran Sasso Science Institute, L'Aquila, Italy}
  \and Yuval Emek\thanks{Technion -- Israel Institute of Technology, Israel}
	\and Christoph Lenzen\thanks{Max Planck Institute for Informatics, Saarbr\"ucken, Germany}
}
\date{}

\begin{document}

\setcounter{page}{0}

\maketitle

\begin{abstract}
  In many models for large-scale computation, decomposition of the problem is
  key to efficient algorithms.
  For distance-related graph problems, it is often crucial that such a
  decomposition results in clusters of small diameter, while the
  probability that an edge is cut by the decomposition scales linearly with the
  length of the edge.
  There is a large body of literature on low diameter graph decomposition with
  small edge cutting probabilities, with all existing techniques heavily
  building on \emph{single source shortest paths (SSSP)} computations.
  Unfortunately, in many theoretical models for large-scale computations, the
  SSSP task constitutes a complexity bottleneck.
  Therefore, it is desirable to replace exact SSSP computations with approximate
  ones.
  However this imposes a fundamental challenge since the existing constructions
  of low diameter graph decomposition with small edge cutting probabilities
  inherently rely on the subtractive form of the triangle inequality, which
  fails to hold under distance approximation.

  The current paper overcomes this obstacle by developing a technique termed
  \emph{blurry ball growing}.
By combining this technique with a clever algorithmic idea of Miller et al.\
(SPAA 2013), we obtain a construction of low diameter decompositions with
small edge cutting probabilities which replaces exact SSSP computations by (a
small number of) approximate ones.
  The utility of our approach is showcased by deriving efficient algorithms that
  work in the \Congest, \Pram, and semi-streaming models of computation.
As an application, we obtain metric tree embedding algorithms in the vein of
Bartal (FOCS 1996) whose computational complexities in these models are
optimal up to polylogarithmic factors.
  Our embeddings have the additional useful property that the tree can be
  mapped back to the original graph such that each edge is ``used'' only
  $O(\log n)$ times, which is of interest for capacitated problems and
  simulating \Congest algorithms on the tree into which the graph is embedded.
\end{abstract}

\thispagestyle{empty}
% \pagebreak

%!TEX root = ./embeddable_tree.tex
%%%%%%%%%%%%%%%%%%%%%%%%%%%%%%%%%%%%%%%%%%%%%%%%%%%%%%%%%%%%%%%%%%%%%%%%%%%%%%
\section{Introduction}
\label{section:introduction}
%%%%%%%%%%%%%%%%%%%%%%%%%%%%%%%%%%%%%%%%%%%%%%%%%%%%%%%%%%%%%%%%%%%%%%%%%%%%%%
Consider an $n$-vertex graph
$G = (V, E, \ell)$,
where
$\ell : E \rightarrow \ZZ_{> 0}$
is an edge \emph{length} function.\footnote{%
We sometimes use the shorthand $\ell_e$ for $\ell(e)$.
}
The \emph{distance} between two vertices $u$ and $v$ in $G$, denoted by
$\dist_{G}(u, v)$,
is defined to be the length with respect to $\ell$ of a shortest
$(u, v)$-path
in $G$.
The \emph{diameter} of $G$ is the maximum distance between any two vertices,
denoted by
$\diam(G) = \max_{u, v \in V} \{\dist_{G}(u, v)\}$.

A \emph{decomposition} $D$ of $G$ is a partition of the vertex set $V$ into
pairwise disjoint \emph{clusters}.
Such a decomposition induces a (multiway) \emph{cut} on $G$ and we use
$\Ecut(D)$ to denote the subset of edges that cross this cut, namely, edges
whose endpoints belong to different clusters of $D$.
The \emph{weight} of the decomposition $D$ is defined to be the sum
$\sum_{e \in \Ecut(D)} \frac{1}{\ell_e}$
of the reciprocal lengths of the edges crossing its cut.
Our focus in this paper is on the construction of decompositions whose
clusters' diameter is bounded by some specified parameter $r$ (the notion of a
cluster's diameter will be made clear soon), referred to hereafter as
\emph{low diameter decompositions}.
The challenging part is to keep the weight of $D$ small.

Low diameter decompositions with small weight were first studied by
Awerbuch~\cite{Awerbuch1985} (see also \cite{AwerbuchP1990, AlonKPW1995}).
Bartal~\cite{Bartal1996} introduced their (combinatorially equivalent)
probabilistic counterpart:
An \emph{$(r, \lambda)$-decomposition} of the graph
$G = (V, E, \ell)$
is a random decomposition $D$ of $G$ such that
(1)
the diameter of each cluster in $D$ is at most $r$;
and
(2)
$\Pr[e \in \Ecut(D)] \leq \frac{\lambda \ell_e}{r}$
for every edge
$e \in E$.
Bartal presented a method that, for a given parameter $r$, constructs an
$(r, O(\log n))$-decomposition and proved the resulting bound on the edge
cutting probabilities to be asymptotically tight.

Low diameter decompositions with small edge cutting probabilities have
proven to be very useful in the algorithmic arena (see
Section~\ref{section:related-work}) and several different techniques have been
developed over the years for constructing them
\cite{AbrahamN2012,Bartal1996,ElkinEST2008,FakcharoenpholRT2004,
DBLP:conf/spaa/MillerPX13, ForsterG2019}.
A common thread of all the existing techniques is that they rely heavily on
making calls to a single source shortest paths (SSSP) subroutine.
While we know how to solve the SSSP problem efficiently in the sequential
(centralized) model of computation, the situation is much more challenging in
restricted models of computation such as the \Congest model of distributed
computing, the parallel random access memory (\Pram) model, or the
semi-streaming graph algorithms model. As it stands, SSSP computations
are the main obstruction to designing efficient constructions of low diameter
decompositions with small edge cutting probabilities in the aforementioned
computational models (and related ones).

%%%%%%%%%%%%%%%%%%%%%%%%%%%%%%%%%%%%%%%
\subsection{Our Contribution}
%%%%%%%%%%%%%%%%%%%%%%%%%%%%%%%%%%%%%%%
In this paper, we introduce a new technique that, given a graph
$G = (V, E, \ell)$
and a parameter $r$, constructs an
$(r, O (\log n))$-decomposition
of $G$.
The crux of our construction is that it does not rely on any \emph{exact} SSSP
computations.
Rather, it efficiently reduces the task to a small number of calls to an
approximate SSSP subroutine.
The technical challenge in this regard stems from the fact that the existing
constructions of low diameter decompositions with small edge cutting
probabilities crucially rely on the subtractive form of the triangle
inequality, stating that
$\dist_{G}(u, v) \geq \dist_{G}(u, w) - \dist_{G}(v, w)$
for every three vertices
$u, v, w \in V$.
Due to the subtraction on the right hand side, the inequality fails if one
replaces exact distances with approximate ones.
The main technical contribution of this paper lies in overcoming this
difficulty.

The approximate SSSP problem can be solved efficiently in the
\Congest \cite{DBLP:conf/wdag/BeckerKKL17},
\Pram \cite{DBLP:journals/jacm/Cohen00},
and
semi-streaming \cite{DBLP:conf/wdag/BeckerKKL17}
models, hence we obtain efficient algorithms for constructing
$(r, O (\log n))$-decompositions
for the three computation models.
These in turn can be invoked recursively to yield efficient
\Congest, \Pram, and semi-streaming constructions of path embeddable trees
\cite{CohenMPPX2014, CohenKMPPRX2014}
and hierarchically well-separated trees
\cite{Bartal1996, Bartal1998, Bartal2004, FakcharoenpholRT2004}
with low \emph{stretch} -- important combinatorial objects in their own
right.
In fact, our low diameter decompositions (and the resulting tree embeddings)
admit an even stronger property.

%###
%%%%%%%%%%%%%%%%%%%%%%%%%%%%%%%%%%%%%%%
\paragraph{Tree-Supported Decompositions.}
%%%%%%%%%%%%%%%%%%%%%%%%%%%%%%%%%%%%%%%
%###
The notion of graph diameter naturally extends from the entire graph
$G = (V, E, \ell)$
to a
vertex subset
$U \subseteq V$
by considering the maximum distance between any two vertices in $U$.
This yields the following distinction:
the \emph{weak diameter} of $U$ in $G$ considers the distances in the
underlying graph $G$, formally defined as
$\max_{u, v \in U} \{\dist_{G}(u ,v)\}$;
the \emph{strong diameter} of $U$ in $G$ considers the distances in the
subgraph $G(U)$ induced by $G$ on $U$, formally defined as
$\diam(G(U))$.\footnote{%
Unless stated otherwise, the edge length function of a subgraph
$H$ of $G$ is the restriction of $\ell$ to $H$'s edge set.
}
In the context of low diameter graph decompositions with small edge
cutting probabilities, both the weak and strong notions of the cluster diameter
have been considered in the literature.
As we now explain, the current paper adopts a diameter notion that falls
somewhere in between the two.

For a decomposition $D$ of the graph
$G = (V, E, \ell)$,
we require that each cluster
$C \in D$
is associated with a tree
$T_C = (U_C, F_C)$,
referred to as the \emph{supporting tree} of $C$,
that is a subgraph of $G$ and spans $C$, i.e.,
$C \subseteq U_C \subseteq V$
and
$F_C \subseteq E$.
To emphasize this requirement, we refer to the decomposition $D$ as a
\emph{tree-supported decomposition (TSD)}.
The \emph{diameter} of a TSD $D$ of $G$ is then defined to be the maximum
diameter of any of its supporting trees, denoted by
$\diam(D) = \max_{C \in D} \{ \diam(T_{C}) \}$.

Notice that if the supporting tree $T_{C}$ of each cluster $C \in D$
is required to be a spanning tree of $G(C)$, then $\diam(D)$ bounds the strong
diameter of $D$'s clusters.
This requirement is not imposed in the current paper, allowing $T_{C}$ to use
edges (and vertices) outside of $G(C)$, meaning that $\diam(D)$
merely bounds the weak diameter of the clusters.
However, we do require that the maximum edge \emph{load} is kept small, where
the load of edge
$e \in E$
in $D$ is defined to be the number of clusters
$C \in D$
such that $e$ is included in the supporting tree of $C$, denoted by
$\Load_{D}(e) = |\{ C \in D : e \in F_{C} \}|$.
The properties of our graph decomposition construction can now be formally
stated.

\begin{theorem} \label{theorem:main}
There exists a (randomized) algorithm that given a graph
$G = (V, E, \ell)$
with $\poly(n)$-bounded edge lengths and a parameter
$r \leq \diam(G)$,
constructs a random TSD $D$ of $G$ with the following guarantees:
(1)
$\diam(D) \leq r$
w.h.p.;\footnote{%
We say that event $A$ occurs \emph{with high probability}, abbreviated
\emph{w.h.p.}, if
$\Pr[A] \geq 1 - n^{-c}$,
where $c$ is an arbitrarily large constant chosen upfront.}
(2)
$\max_{e \in E} \{\Load_{D}(e)\} \leq O (\log n)$
w.h.p.;
and
(3)
$\Pr \left[ e \in \Ecut(D) \right]
\leq
O \left( \frac{\ell_{e} \cdot \log n}{r} \right)$
for every edge
$e \in E$.
\end{theorem}

The algorithm promised in Theorem~\ref{theorem:main} is based on combining a
novel technique termed \emph{blurry ball growing} with the algorithmic ideas
of Miller et al.~\cite{DBLP:conf/spaa/MillerPX13}.
As discussed earlier, this combination allows us to implement our algorithm
using an approximate SSSP subroutine (without any exact SSSP computations).
By example of the \Congest, \Pram, and semi-streaming models, we show that
this leads to efficient implementations.
We stress that what little computation is performed beyond approximate SSSP
computations is very easy, if not trivial, to implement.
Accordingly, we expect the technique to carry over to further computational
models.

We emphasize that our decomposition maintains a
small load of
$O(\log n)$
on the edges.
Consequently, in many situations, our decomposition can be used in an
identical way as a strong diameter decomposition with only polylogarithmic
overheads.
For example, although we cannot construct low average stretch spanning trees
as these are required to be subgraphs of the original graph, we can construct
\emph{projected trees} (see Section~\ref{section:projected-tree}), a special
case of path-embeddable trees
\cite{CohenKMPPRX2014, CohenMPPX2014}.
Projected trees have a mapping of their edges to the original graph such that,
e.g., a \Congest algorithm on the projected tree can be simulated on the
original graph at an
$O(\log n)$
overhead in round complexity.
Our result is related to the low-congestion shortcuts of Ghaffari and Haeupler~\cite{ghaffariH16} with the following differences. In Ghaffari and Haeupler's work the partition is chosen by an adversary and the input is restricted to unweighted graphs. In contrast, our technique constructs the partition, but weighted graphs can be treated as well.
% Thus, in essence, we obtain a weighted counterpart to the low-congestion
% shortcuts of Ghaffari and Haeupler~\cite{ghaffariH16}.
A further possible application of our projected trees is in the field of
solvers for symmetric diagonally dominant linear systems, utilizing them in a
similar way as low average stretch spanning trees
(cf.~\cite{CohenMPPX2014, CohenKMPPRX2014}).
Prior algorithms for metric tree embeddings lack this property and,
accordingly, cannot take this role.

%%%%%%%%%%%%%%%%%%%%%%%%%%%%%%%%%%%%%%%
\subsection{Structure of this Paper}
%%%%%%%%%%%%%%%%%%%%%%%%%%%%%%%%%%%%%%%
We first fix some notation and state basic facts in the preliminaries in
Section~\ref{section:preliminaries}. In Section~\ref{sec: blur}, we present
the blurry ball growing technique that we use in
Section~\ref{section:graph-decomposition} in order to obtain the routine for
computing a random TSD of low diameter, load, and edge
cutting probability, as promised in Theorem~\ref{theorem:main}.
In Section~\ref{section:virtual-trees}, we highlight some applications of this
routine:
We first explain how to obtain a hierarchical decompositions by applying
the method recursively (Section~\ref{section:hierarchical-decompositions}) and
then show how to obtain random projected trees
(Section~\ref{section:projected-tree}) and hierarchically well-separated trees
(Section~\ref{section:hst}) with
$O (\log^{2} n)$
bound on the expected stretch.
We also show that this bound can be improved to
$O (\log n)$
by considering the relaxed notion of $p$-stretch
\cite{CohenKMPPRX2014,CohenMPPX2014} (Section~\ref{section:p-stretch}).
In Section~\ref{sec:implementation}, we explain how to implement our
algorithms in the \Congest, \Pram, and semi-streaming models.
Further related work is reviewed in Section~\ref{section:related-work}.

%%% Local Variables:
%%% mode: latex
%%% TeX-master: "embeddable_tree"
%%% End:

%!TEX root=./embeddable_tree.tex
%###
\section{Preliminaries}
%###
\label{section:preliminaries}
We start with basic notation.
We consider a weighted, undirected, connected $n$-vertex graph $G=(V, E, \ell)$,
where $\ell:E\rightarrow \{1,2,\ldots,n^{O(1)}\}$ is an edge length function.%
\footnote{All graphs in this paper are assumed to be finite, undirected, and
connected. The assumption of integral edge weights is made for convenience; it
suffices if the aspect ratio $\frac{\max_{e\in E}\{\ell_e\}}{\min_{e\in
E}\{\ell_e\}}=n^{O(1)}$.} For a subgraph $H$ of $G$,
we denote by $\dist_H(u,v)$ the length of the shortest path
between two nodes \(u\) and \(v\) in $H$.
If $H=G$, we may omit the subscript.
For a set \(B\subseteq V\) and a node \(v\in V\),
we use
\(\dist(B, v):=\min_{u\in B}\{\dist(u, v)\}\)
to denote the distance of the node \(v\) to the set \(B\).
For a set of vertices $U\subseteq V$, we denote by
$\Ecut(U):=\{e=\{u,v\}\in E\colon u\in U, v\in V\setminus U\}$
the set of edges that are ``cut'' by $U$.

%###
\paragraph{Approximate Single Source Shortest Paths.}
%###
The main subroutine we use in our approach are
$(1+\eps)$-\emph{approximate} SSSP computations for undirected graphs.
A $(1+\eps)$-\emph{approximate SSSP algorithm} is an algorithm that
takes as input a weighted undirected graph $G=(V, E, \ell)$ and a source node
$s\in V$ and returns
% a distance function $d:V\rightarrow \RR_{\ge 0}$
% and
a spanning tree $T$ of $G$ such that, for every node $v\in V$,
the length of the path from $s$ to $v$ in $T$ is at most $(1+\eps)\cdot \dist(s,
v)$, i.e., $\dist(s,v)\leq \dist_T(s,v)\leq (1+\eps)\cdot \dist(s,v)$.
% We remark that, strictly speaking, our methods do not require
% the $(1+\eps)$-\emph{approximate} SSSP algorithm to return the entire tree $T$,
% but the following information suffices: for every node $v\in V$,
% besides the approximate distance $d(v)$,
% we need the first node on the path from
% $s$ to $v$, i.e., instead of the whole tree $T$,
% the first level of $T$ suffices.

%###
\paragraph{Super-Source Graphs.}
%###
Our approach requires $(1+\eps)$-approximate SSSP computations
in graphs $G_s$ that result from
subgraphs of $G$ by adding a (virtual) \emph{super-source node} $s\notin V$:
\begin{definition}[Super-source graphs]\label{def:super}
  Fix a subgraph $H=(V_H,E_H,\ell|_H)$ of $G$.
  Construct $G_s=(V_H\dot{\cup}\{s\},E_H\cup E_s,\ell^{G_s})$ by choosing
  $E_s\subseteq V_H\times \{s\}$, picking $\ell^{G_s}_e\in \{1,\ldots,n^c\}$
  for $e\in E_s$, and setting $\ell^{G_s}_e=\ell_e$ for all $e\in E_H$.
  We refer to $G_s$ as a \emph{super-source graph} (of $G$) and to $s$ as its
  \emph{super-source.}
\end{definition}
We note that one way of obtaining a super-source graph of a graph $G$ is
to contract a subset of nodes, say $B$, into a super-source $s$.
In this case $V_H=V\setminus B$ and the edges $E_s$ and their lengths
result from the contraction of $B$ into $s$.

%###
\paragraph{Exponential Distribution.}
%###
We denote the exponential distribution with mean \(\frac{1}{\beta}\) by \(\Exp_\beta\).
Using the Heaviside step function that is defined as
$H(x)=0$ if $x < 0$ and
$H(x)=1$ otherwise,
the density function of the exponential distribution is given by
\(f_{\Exp_\beta}(x)=\beta \exp(-\beta x)\cdot H(x)\).
Its cumulative density function is
\(F_{\Exp_\beta}(x)=(1 - \exp(-\beta x))\cdot H(x)\).
A standard result is that drawing from this distribution results in values of
$O(\beta \log n)$ w.h.p.:

\begin{lemma}\label{lemma:radius_bound}
  For parameters $0<\varepsilon<1$, $\beta>0$,
  and a sufficiently large constant $c>0$,
  let $t:=\frac{c\log n}{4(1+\eps)\beta}$ and $X\sim \Exp_{\beta}$.
  Then $P[X\geq t]= n^{-\Omega(c)}$, i.e., $X<t$ w.h.p.
\end{lemma}
\begin{proof}
  Using the form of the density function, we get
	\begin{equation*}
		P[X\geq t]
    =\frac{\int_{t}^{\infty}\exp(-\beta x)\,dx}
          {\int_{0}^{\infty}\exp(-\beta x)\,dx}
		=\frac{\exp(-\beta t)\int_{0}^{\infty}\exp(-\beta x)\,dx}
          {\int_{0}^{\infty}\exp(-\beta x)\,dx}
		= \exp(-\Omega(c\log n)) = n^{-\Omega(c)}\,.\qedhere
\end{equation*}
\end{proof}

We will make heavy use of the following lemma,
see the paper by Miller et al.~\cite{DBLP:conf/spaa/MillerPX13} for the proof.
Note that in their paper they state the lemma with an upper bound of
$O(\beta c)$ on the probability, although their proof in fact
bounds the probability by exactly $\beta c$.
\begin{lemma}
  [Lemma 4.4 in~\cite{DBLP:conf/spaa/MillerPX13}]
  \label{lemma:lemma4.4}
	Let \(d_1\le \ldots\le d_s\) be arbitrary values and
  \(\delta_1,\ldots, \delta_s\)
  be independent random variables picked from \(\Exp_\beta\).
  Then the probability that the smallest and the second smallest values of
  \(d_i-\delta_i\) are within \(c\) of each other is at most \(\beta c\).
\end{lemma}

Miller et al.~\cite{DBLP:conf/spaa/MillerPX13} used this lemma to analyze
the following ball growing technique that proceeds in time steps.
Every node $u$ in the graph grows a ball $B_u$ independently
and in parallel,
but with a delay of $\delta_u$ time steps, where $\delta_u\sim\Exp_\beta$.
Every ball increases its radius by 1 in each time step
and we say that the ball $B_v$ ``arrives'' at node $u$,
if node $v$ minimizes $\dist(u, v) - \delta_v$ over all nodes.
In this case $u$ ``gets absorbed'' by $v$'s ball $B_v$.
The process stops when every node $u$ is absorbed by some ball.
Notice that $u$ gets absorbed by its own ball $B_u$,
if and only if no other ball arrives at $u$ during the first
$\delta_u$ time steps.

Now consider an arbitrary edge $e$ in the graph
and imagine it to be split into two equal length edges by a node $v_e$.
If we let $d_1\le \ldots \le d_n$ denote the $n$ values
$\dist(u, v_e) - \delta_u$ for every $u\in V$,
the above lemma shows that the arrival times of the first and second ball
at node $v_e$ differ by at least $2\ell_e$
with probability $1-O(\beta \ell_e)=1-O(\frac{\ell_e\log n}{\eps r})$,
when choosing $\beta=\Theta(\frac{\log n}{\eps r})$.
Hence the lemma allows for bounding the probability of an edge being cut
by such ball growing process with exponentially distributed delays.

We remark that the implementations in Section~\ref{sec:implementation} draw from
discrete distributions. Rounding continuous distributions to multiples of
$n^{-c}$ for sufficiently large $c\in O(1)$ yields w.h.p.\ the same results, but
limits the number of random bits required to draw and store a random value to
$O(\log n)$.

%!TEX root=./embeddable_tree.tex
\section{Blurry Ball Growing}
\label{sec: blur}
In this section,
we describe a routine \texttt{blur} that takes as input
a graph $G$,
a node set $B\subseteq V$,
and parameters $\rho$ and $\alpha$
and outputs a superset $U$ of $B$.
It guarantees that nodes in $U$ are not too far from $B$, yet the probability to
cut edges is small.
More precisely, we show the following theorem.
\begin{theorem}\label{thm:blur}
  Let $n\geq 2$ and $\alpha=\frac{1}{2\log n}$. There is a routine
  \textnormal{\texttt{blur}}$(G,\rho,B,\alpha)$ that outputs
  a superset $U$ of $B$ such that:
  \begin{compactenum}
    \item\label{item: prop 1} For every edge $e\in E$,
    the probability that $e\in \Ecut(U)$ is bounded by $O(\frac{\ell_e}{\rho})$.
    \item\label{item: prop 2} For every $v\in U$,
    it holds that $\dist(B,v)\leq \frac{\rho}{1-\alpha}$.
  \end{compactenum}
\end{theorem}
The routine \texttt{blur}, see Algorithm~\ref{alg:blur}, is based on
$(1+\eps)$-approximate SSSP computations and contractions of node sets and thus
can be readily parallelized. The basic idea is to grow a ball of uniformly
random radius around $B$, where contraction of $B$ yields the super-source of
the SSSP computation. However, as approximating distances may imply that the
``noise'' due to the relative $\eps$-error may cut a short edge with a
comparatively large probability, the procedure is repeated with random radii
drawn from uniform distributions with width that decrease by factor $\alpha$ in
each step. To make this work, the approximation error of the SSSP algorithm must
satisfy $\varepsilon\leq \alpha^2$. Accordingly, it would be desirable to chose
$\alpha$ large for the sake of small computational costs in the approximate SSSP
routine. However, it turns out that, in order to achieve Property~\ref{item:
prop 1} in Theorem~\ref{thm:blur}, $\alpha$ has to satisfy $\alpha=O\big(\frac{\log \log
n}{\log n}\big)$.

\begin{algorithm}
	\label{alg:blur}
	\caption{\texttt{blur}\((G, \rho,  B, \alpha)\)}
	\Input{graph $G=(V,E,\ell)$,
    positive $\rho$,
    set $B\subset V$,
    positive $\alpha\leq \frac{1}{2}$}
	\Output{set of nodes \(U\)}
	\medskip
	\(i:=0\), $B^{[0]}:=B$\;
	\While{\(\alpha^i\rho \ge \min_{e\in E}\{\ell_e\}\)}{
		\(i:=i+1\), \(r^{[i]}\in \UUU[0, \alpha^{i-1} \rho]\).\;
		Obtain super-source graph \(G^{[i]}\) from \(G\)
    by contracting \(B^{[i-1]}\) into a super-source node \(s^{[i]}\)\;
		Compute $(1+\alpha^2)$-approximate SSSP tree $T^{[i]}$ of $G^{[i]}$\;
		\(
      B^{[i]}
      :=B^{[i-1]}\cup\{v\in G^{[i]}
        \,|\,\dist_{T^{[i]}}(s^{[i]},v)\le r^{[i]}\}\setminus \{s^{[i]}\}
    \)\;
	}
	\Return $\bigcup_{j=0}^i B^{[i]}$
\end{algorithm}

%###
\paragraph{Analysis.}
%###
We begin with Property~\ref{item: prop 2}, which readily follows from sampling
$r^{[i]}$ from $\UUU[0, \alpha^{i-1} \rho]$.
\begin{lemma}\label{lemma:reachable}
	If
  $
    \dist_{G^{[i+1]}}(s^{[i+1]}, u)
    =\dist(B^{[i]}, u)
    \ge \frac{\alpha^{i} \rho}{1-\alpha}
  $
  for some $i$, then $u\notin U$.
  In particular, it holds that $\dist_G(B, v)\le \frac{\rho}{1-\alpha}$
  for every $v\in U$.
\end{lemma}
\begin{proof}
  Any $u\in B^{[k]}$ for $k>i$ has distance to $B^{[i]}$
  at most
  $\sum_{j \ge i+1} r^{[j]}
  	\le \sum_{j \ge i+1} \alpha^{j-1} \rho
  	< \alpha^{i} \rho \sum_{j =0}^\infty \alpha^j
  	= \frac{\alpha^{i} \rho}{1-\alpha}$, showing the first claim.
  Setting $i=0$ yields the second claim.
\end{proof}

It remains to verify Property~1, i.e., that the probability of cutting an
arbitrary edge \(e\in E\) is \(O(\frac{\ell_e}{\rho})\).
We start with the following definition.
\begin{definition}\label{def:safe}
	We say that \(\{u,v\}\in E\) \emph{is safe after step \(i\)}
  of {\normalfont\texttt{blur}\((G, \rho, B, \alpha)\)},
  if either \(u,v\in B^{[i]}\) or
  \(
    \min\{\dist_{G^{[i+1]}}(s^{[i+1]}, u),\dist_{G^{[i+1]}}(s^{[i+1]}, v)\}
    \ge \frac{\alpha^{i}\rho}{1-\alpha}.
  \)
\end{definition}

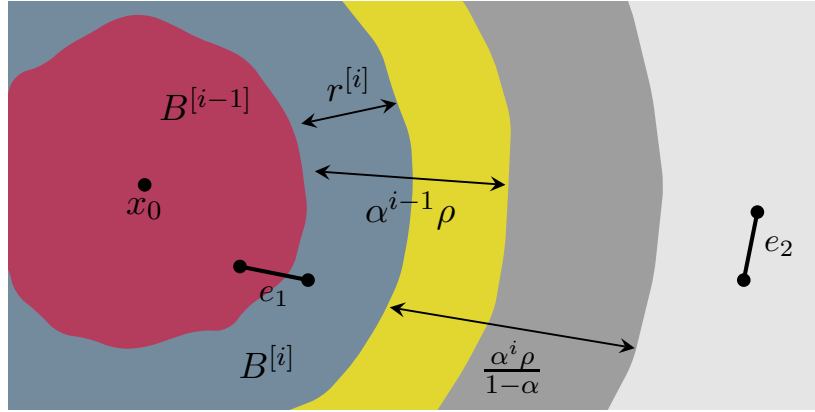
\begin{figure}[ht!]
	\begin{center}
		\resizebox{0.7\textwidth}{!}{
			%!TEX root=./embeddable_tree.tex

\begin{tikzpicture}
\tikzstyle{every node}=[scale=.7]
  \clip (-1,-1.65) rectangle + (6, 3.0);

  \fill[gray, opacity=0.2] (-10, -10) rectangle + (20, 20);

  \coordinate (c) at (0, 0);
  % \fill[gray, opacity=0.7] (0, 0) circle (3.8) node[black, opacity=1] at (0, 0) {};
  % \fill[yellow, opacity=0.7] (0, 0) circle (2.7) node[black, opacity=1] at (0, 0) {};
  \fill[gray, rounded corners=1mm, opacity=0.7] (c) \irregularcircle{3.8cm}{.2mm} node[black, opacity=1] at (0.9, -1.3) {};`
  \fill[yellow, rounded corners=1mm, opacity=0.7] (c) \irregularcircle{2.7cm}{.4mm} node[black, opacity=1] at (0.9, -1.3) {};`
  \fill[hanblue, rounded corners=1mm, opacity=0.7] (c) \irregularcircle{2cm}{.55mm} node[black, opacity=1] at (0.9, -1.3) {$B^{[i]}$};`
  \fill[alizarin, rounded corners=1mm, opacity=0.7] (c) \irregularcircle{1.2cm}{.7mm} node[black, opacity=1] at (0.45, 0.6) {$B^{[i-1]}$};

  \draw[black, <->, >=stealth] (1.15, .45) -- (1.85, 0.6) node[above, midway] {$r^{[i]}$};
  \draw[black, <->, >=stealth] (1.25, 0.1) -- (2.65, -0.0) node[below, midway] {$\alpha^{i-1} \rho$};
  \draw[black, <->, >=stealth] (1.8, -.9) -- (3.6, -1.2) node[below, midway] {$\frac{\alpha^{i} \rho}{1-\alpha}$};

  \fill[black] (0,0) circle (.5mm) node[below] {$x_0$};

  \draw[fill, black, line width=.3mm] (1.2,-.7) node {} -- (0.7,-.6) node {} node[below, midway] {\small $e_1$};
  \fill[black] (1.2,-.7) circle (.5mm);
  \fill[black] (.7,-.6) circle (.5mm);

  \draw[fill, black, line width=.3mm] (4.4,-.7) node {} -- (4.5,-.2) node {} node[right, midway] {\small $e_2$};
  \fill[black] (4.4,-.7) circle (.5mm);
  \fill[black] (4.5,-.2) circle (.5mm);
\end{tikzpicture}
		}
	\end{center}
	\caption{
		An illustration of the blurry ball growing procedure
    \texttt{blur}\((G, \rho,  B, \alpha)\) in iteration $i$.
    The radius $r^{[i]}$ is sampled uniformly from $[0,\alpha^{i-1}\rho]$
    and $B^{[i]}$ is defined as all nodes whose
    $(1+\alpha^2)$-approximate distance to $B^{[i-1]}$ is at most $r^{[i]}$.
		Both edges $e_1$ and $e_2$ are safe from being cut after iteration $i$:
    $e_1$ has both endpoints in $B^{[i]}\subseteq U$
    and both endpoints of $e_2$ are farther away from $B^{[i]}$
    than $\frac{\alpha^i\rho}{1-\alpha}$,
    meaning that neither of them will lie in $U$ after termination.
	}
	\label{fig: blurry ball}
\end{figure}
Clearly, if \(\{u,v\}\in E\) is safe after step \(i\)
of \texttt{blur}\((G, \rho,  B, \alpha)\),
then \(e\notin \Ecut(U)\):
if $u,v\in B^{[i]}$, then $u,v\in U$ by construction;
if
$
  \min\{\dist_{G^{[i+1]}}(s^{[i+1]}, u),\dist_{G^{[i+1]}}(s^{[i+1]}, v)\}
  \geq \frac{\alpha^{i}\rho}{1-\alpha},
$
it follows that $u,v\notin U$ by Lemma~\ref{lemma:reachable}.
See Figure~\ref{fig: blurry ball} for an illustration of these two events.
Thus, in order to bound the probability of an edge being cut,
it suffices to bound the probability that an edge never becomes safe.
Accordingly, we define \(X_{i,e}\) to be the event that \(e\) is not safe
after step \(i\) of the algorithm
conditioned on the event that \(e\) was not safe after step \(i-1\) and bound
$P[X_{i,e}]$.
\begin{lemma}\label{lemma:edge stays unsafe}
	For each iteration $i$ and edge $e\in E$, it holds that
  \(
    \Pr[X_{i,e}]
    \le
    \frac{5}{4} \cdot \frac{\ell_e}{\alpha^{i-1} \rho}
    + \alpha \cdot (1+4\alpha).
  \)
\end{lemma}
\begin{proof}
	By Definition~\ref{def:safe}, it follows that
  if \(e=\{u,v\}\) is not safe after step \(i\),
  we must, w.l.o.g.\ over the choice of $u, v$,
  have \(\dist_{G^{[i+1]}}(s^{[i+1]}, u)< \frac{\alpha^{i}\rho}{1-\alpha}\)
  as well as $u,v\not \in B^{[i]}$.
  By the approximation guarantee of the SSSP algorithm
  and the triangle inequality, the latter entails that
  \begin{align*}
    r^{[i]}<\max\{\dist_{T^{[i]}}(s^{[i]}, u), \dist_{T^{[i]}}(s^{[i]}, v)\}
    &\leq (1+\alpha^2)\max\{\dist_{G^{[i]}}(s^{[i]}, u),
                            \dist_{G^{[i]}}(s^{[i]}, v)\}\\
    &\leq (1+\alpha^2)(\dist_{G^{[i]}}(s^{[i]}, u)+\ell_e).
  \end{align*}
  From the former inequality, we get that
	\begin{equation}\label{formula:di_ri}
		\dist_{G^{[i]}}(s^{[i]}, u)
		\le \dist_{G^{[i+1]}}(s^{[i+1]}, u) + r^{[i]}
		< \frac{\alpha^{i}\rho}{1-\alpha} + r^{[i]},
	\end{equation}
  which yields that
  \(
    r^{[i]}
	   \ge \dist_{G^{[i]}}(s^{[i]}, u) - \frac{\alpha^{i}\rho}{1-\alpha}.
  \)
  As $r^{[i]}$ is drawn uniformly from an interval of length
  $\alpha^{i-1}\rho$,
  these lower and upper bounds on
  \(r^{[i]}\)
  readily imply a bound on the probability of $X_{i,e}$:
	\begin{align}
			\Pr[X_{i,e}]
			&\le \Pr\Big[r^{[i]}\in
        \Big(
          \dist_{G^{[i]}}(s^{[i]}, u) - \frac{\alpha^{i}\rho}{1-\alpha},
          (1+\alpha^2)\cdot (\dist_{G^{[i]}}(s^{[i]}, u) + \ell_e)
        \Big)
        \Big]\nonumber\\
			&\le \frac{(1+\alpha^2)\ell_e}{\alpha^{i-1}\rho}
        + \frac{\alpha^2 \dist_{G^{[i]}}(s^{[i]}, u)}{\alpha^{i-1}\rho}
        + \frac{\alpha}{1-\alpha}.
      \label{formula:prob xie}
	\end{align}
	Moreover, from~\eqref{formula:di_ri} and \(r^{[i]}\le \alpha^{i-1}\rho\),
  we conclude that
  \(
    \dist_{G^{[i]}}(s^{[i]}, u)
    < \alpha^{i-1} \rho\cdot (1+\frac{\alpha}{1-\alpha})
    = \frac{\alpha^{i-1} \rho}{1-\alpha}.
  \)
	Plugging into~\eqref{formula:prob xie}, with $\alpha\leq \frac{1}{2}$ we get
	that $\Pr[X_{i,e}]
    \leq \frac{5}{4}\cdot\frac{\ell_e}{\alpha^{i-1}\rho}
      + \frac{\alpha^2}{1-\alpha} + \frac{\alpha}{1-\alpha}
    \leq \frac{5}{4}\cdot\frac{\ell_e}{\alpha^{i-1}\rho}
      + \alpha(1+4\alpha)$.
\end{proof}
Applying this lemma to all iterations
in which $e$ has a significant probability to become safe
(i.e., all iterations $i$ for which $\alpha^{i-1}r \geq \ell_e$),
we obtain the desired bound on the probability that $e$ is cut.
\begin{lemma}\label{lemma:prob_cut}
	If $\alpha =O\left(\frac{\log\log n}{\log n}\right)$, then $\Pr[e\in \Ecut(U)] = O\left(\frac{\ell_e}{\rho}\right)$ for each $e\in E$.
\end{lemma}
\begin{proof}
	If $\ell_e> \rho$, trivially $\Pr[e\in \Ecut(U)]\leq 1 < \frac{\ell_e}{\rho}$.
  Otherwise, we let $i_e\geq 1$ be the largest index such that
  \(\ell_e\le \alpha^{i_e-1} \rho\).
	By Lemma~\ref{lemma:edge stays unsafe},
  for all $i$ the probability that an edge that is not safe
  after $i-1$ steps is still not safe after step $i$ is bounded by
	$\Pr[X_{i,e}]
    \le\frac{5}{4} \cdot \frac{\ell_e}{\alpha^{i-1} \rho}
      + \alpha \cdot (1+4\alpha)$.
  Depending on the index $i$, we differentiate this upper bound further:
  \begin{compactitem}
    \item Case $i=i_e$:
    As $\alpha^{i_e}\rho<\ell_e$,
    we get that $\alpha<\frac{\ell_e}{\alpha^{i_e-1}\rho}$.
    With $\alpha\leq \frac{1}{2}$,
    $\Pr[X_{i_e,e}]<\frac{5 \ell_e}{\alpha^{i_e-1}\rho}$ follows.
    \item Case $i=i_e-1$:
    Then
    \(
      \frac{\ell_e}{\alpha^{i-1} \rho}
      \le \frac{\alpha^{i_e-1}}{\alpha^{i_e-2}}
      = \alpha,
    \)
    yielding with $\alpha \leq \frac{1}{2}$ that $\Pr[X_{i_e-1,e}]< 5\alpha$.
    \item Case $i\leq i_e-2$: This entails that
    \(\frac{\ell_e}{\alpha^{i-1}\rho}\le \alpha^2\) and thus
    $
      \Pr[X_{i,e}]
      < 2\alpha^2 + \alpha\cdot (1+4\alpha)= \alpha\cdot (1+6\alpha).
    $
  \end{compactitem}
  Using these bounds and distinguishing cases based on $i_e$,
  we can bound the overall probability that the edge is cut.
  \begin{compactitem}
    \item Case $i_e=1$:
    $\Pr[e\in \Ecut(U)] \le \Pr[X_{1,e}]=\Pr[X_{i_e,e}]<\frac{5\ell_e}{\rho}$.
    \item Case $i_e=2$:
    $
      \Pr[e\in \Ecut(U)]
      \le \Pr[X_{2,e}]\cdot \Pr[X_{1,e}]
      = \Pr[X_{i_e,e}]\cdot \Pr[X_{i_e-1,e}]
      < \frac{5\ell_e}{\alpha\rho}\cdot 5\alpha
      = \frac{25\ell_e}{\rho}.
    $
    \item Case $i_e\geq 3$:
  	$\Pr[e\in \Ecut(U)]
  		\le \frac{25\ell_e}{\rho}\cdot
              \prod_{i\le i_e-2}\Pr[X_{i,e}]
  		\le \frac{25\ell_e}{\rho}\cdot (\alpha(1+6\alpha))^{i_e-2}
  		< \frac{25\ell_e}{\rho} \cdot (1+6\alpha)^{i_e}$.
  \end{compactitem}
  Hence, it remains to bound $(1+6\alpha)^{i_e}= O(1)$ to complete the proof.
  Noting that
  $
    i_e
    =\left\lceil \log_{1/\alpha}\frac{\rho}{\ell_e}\right\rceil
    = O\left(\frac{\log n}{\log (1/\alpha)}\right)
  $
  due to the assumption that edge lengths are from $1,\ldots,n^{O(1)}$,
  we have that
  \begin{equation*}
    (1+6\alpha)^{i_e}
    = (1+6\alpha)^{O(\log n /\log (1/\alpha))}
    = \left((1+6\alpha)^{1/(6\alpha)}\right)^{O(\alpha\log n /\log (1/\alpha))}
    =e^{O(\alpha\log n /\log (1/\alpha))}
  \end{equation*}
  and therefore the precondition that
  $\alpha=O(\frac{\log \log n}{\log n})$ implies the statement of the lemma.
\end{proof}
Theorem~\ref{thm:blur} now follows from Lemmas~\ref{lemma:reachable}
and~\ref{lemma:prob_cut}.

%!TEX root=./embeddable_tree.tex
%%%%%%%%%%%%%%%%%%%%%%%%%%%%%%%%%%%%%%%%%%%%%%%%%%%%%%%%%%%%%%%%%%%%%%%%%%%%%%
\section{Tree-Supported Decomposition}
\label{section:graph-decomposition}
%%%%%%%%%%%%%%%%%%%%%%%%%%%%%%%%%%%%%%%%%%%%%%%%%%%%%%%%%%%%%%%%%%%%%%%%%%%%%%
In this section, we present the construction of TSDs that admit low diameter,
low load, and low edge cutting probability, establishing
Theorem~\ref{theorem:main}.
Our method is inspired by the partition technique
from~\cite{DBLP:conf/spaa/MillerPX13} that allows for efficient
parallel and distributed implementations. However, we seek to
rely on approximate rather than on exact distance computations.

To motivate our approach, consider naive application of the decomposition
technique from~\cite{DBLP:conf/spaa/MillerPX13} using approximate rather than
exact distance computations. This would look as follows:
One would add a super-source $s$ to the graph,
assign exponentially sampled lengths to the edges adjacent to $s$,
compute $(1+\eps)$-approximate distances
from $s$ to all nodes for some small enough $\eps$,
and assign nodes to the root of the subtree of $s$ that they are situated in.
This approach certainly leads to a decomposition of $G$.
However, a consequence of the approximate distance computation is that
the probability to cut a short edge is dominated by the approximating error,
which is $\eps$ times the distance to the source, which may be very large
compared to the length of the edge.

In order to still ensure the desired bound, we seek to employ the blurring
technique from the previous section to clusters obtained as described above.
This introduces the new obstacle that the clusters need to be separated from
each other first, as the blurring procedure grows the clusters by a random
radius. We enforce this separation by removing from each cluster every node
that is too close to its boundary; Property~\ref{item: prop 2} of
Theorem~\ref{thm:blur}, namely $\dist_G(B,v)\leq \frac{\rho}{1-\alpha}$ for
blurring cluster $B$, determines what precisely is ``too close.''
While this may result in a large portion of the graph not being contained
in any cluster even after blurring all clusters,
we can ensure that each edge is contained in some cluster with at least
probability $p=\Omega(1)$ (or is very long and can be safely deleted).
Repeating the procedure $O(\log n)$ times hence completes the
decomposition w.h.p.
\begin{algorithm}[hbt]
	\label{alg:decomposition}
	\caption{\decomp$(G, \Delta)$}
	\Input{graph \(G=(V,E,\ell)\) and $\Delta\in \NN$}
	\Output{decomposition $D=(C_1,\ldots, C_k)$ of $G$, trees $\TTT=(T_1,\ldots, T_k)$ of depth $\le \frac{\Delta}{2}$ s.t.\ $T_i$ spans a superset of $C_i$}

	\medskip
	\(\beta:= \frac{c \log n}{\Delta}\), $\eps := \frac{1}{c\log^2 n}$, \(D:=\emptyset\), \(\TTT:=\emptyset\)  \atcp{$c$ is a sufficiently large constant}\label{line:beta}
	delete all edges $e\in E$ of length $\ell_e>\frac{1}{40\beta}$\;
	\While{$E(G)\neq \emptyset$}{

		\medskip
		\tcp{\textit{* initial decomposition by exponential shifts *}}
		pick \(\delta_u\sim \Exp_\beta\) for each $u\in V$ independently\label{line:shifts}\;
		\(G_s:=\) super-source graph of $G$ with edges $\{u,s\}$ of length $\ell_{us}=1+\max_{v\in V}\{\delta_v\}-\delta_u$ for $u\in V$\;
		$T:=(1+\eps)$-approximate SSSP tree  for $G_s$ with source $s$\;
		\(R:=\) roots of $T\setminus \{s\}$ and \(\VVV:=(V_u)_{u\in R}\), where $V_u$ are the nodes in $u$'s subtree\label{line:trees}\;

		\medskip
		\tcp{\textit{* separate cells *}}
		$\partial \VVV:=\bigcup_{u\in R}\{v\in V_u\,|\,\exists \{v,w\}\in E\colon w\notin V_u\}$\label{line:shrink_start}\;
		\(G_s':=\) super-source graph of $G$ with edges $\{u,s\}$ of length $1$ for $u\in \partial \VVV$ \;
		$T':=(1+\eps)$-approximate SSSP tree for $G_s'$ with source $s$\;
		\For{each $u\in R$}{
			$V_u^\circ:=V_u\setminus \{v\in V_u\,|\,\dist_{T'}(s,v)\leq \frac{1+\eps}{4\beta}\}$\label{line:shrink_end}\atcp{$V_u^\circ$ is the interior of cell $V_u$}
			$C_u:=$ \texttt{blur}\((G, \rho, V_u^\circ, \frac{1}{2\log n})\), where $\rho:=\frac{1-\frac{1}{2\log n}}{4\beta}$\label{line:blur}\;
			append $C_u$ to $D$ and the subtree of $T$ rooted at $u$ to $\TTT$\;
			$G:=G\setminus C_u$\;
		}
	}
	\Return \((D, \TTT)\)
\end{algorithm}

%###
\paragraph{Algorithm.}
%###
The pseudocode of our procedure \decomp is given
in Algorithm~\ref{alg:decomposition}.
The value $\beta$ chosen in Line~\ref{line:beta} is the parameter
chosen for the exponential distributions:
up to normalization, the density of the distribution is
$\exp(-\beta x)$. The diameter of each (initial) cluster is bounded by
$\max_{v\in V}\{\delta_v\}$,
which we need to be smaller than $\frac{\Delta}{2}$ w.h.p.
However, the probability to cut edges increases as we make the
distributions ``narrower,'' i.e., $\beta$ larger.
Accordingly, we choose $\beta= \Theta\left(\frac{\log n}{\Delta}\right)$,
just small enough to ensure $\delta_v\leq \frac{\Delta}{2}$
w.h.p.\ for all $v\in V$.

The partition from~\cite{DBLP:conf/spaa/MillerPX13} can be interpreted as a
Voronoi decomposition in which each cell center $x_v$ is a virtual copy of its
corresponding node $v\in V$ that is attached to $v$ by an edge of length
$\max_{w\in V}\{\delta_w\}-\delta_v$. Note that the children of the virtual node
$s$ in the (approximate) shortest path tree $T$ are exactly the nodes which have
not been ``absorbed'' into another node's Voronoi cell before they started to
grow their own. Lines~\ref{line:shrink_start} to \ref{line:shrink_end} remove
from each cluster nodes that are in distance (roughly) $\frac{1}{4\beta}$ from
the boundary of the Voronoi cell containing them. Choosing a distance of
$O\left(\frac{1}{\beta}\right)$ here ensures a constant probability that edges
of this length remain in a shrunk cluster; longer edges can safely be cut, as
the required bound on the probability for cutting them is trivial (i.e., $1$), which
is why they are removed at the start of the routine. We then proceed to applying
the blurring subroutine to each (remaining) shrunk cluster. Note that, as the
clusters remain separated due to the choice of parameters, we can realize this
step concurrently for all clusters. The algorithm iterates until all nodes are
assigned to clusters, which requires $O(\log n)$ loop iterations w.h.p.

The remainder of this section is dedicated to proving
Theorem~\ref{theorem:main}.
% We first show the bound on the number of iterations.
% Then, in order to establish the diameter bound of Theorem~\ref{theorem:main},
% we infer Lemma~\ref{lemma:subset}, which also allows us to conclude that
% \decomp in fact outputs a partition of $V$.
% We then proceed to establishing the edge cutting probability bound of
% Theorem~\ref{theorem:main} and finally, the load bound follows
% easily.

%###
\paragraph*{Number of Iterations.}
%###
We first prove the key statement that,
with at least constant probability, for any node $w$,
a ball of radius $\Theta(\frac{1}{\beta})$ around it is contained
within the interior of a cell.
\begin{lemma}\label{lemma:progress}
  Consider an iteration of the while loop of Algorithm~\ref{alg:decomposition} and (by slight abuse of notation) denote by $G=(V,E)$ the subgraph that remains at the beginning of the iteration.
  For any $w\in V$, with at least constant probability a ball of radius $\frac{1}{40\beta}$ around it is contained in the interior of a cell computed in Line~\ref{line:shrink_end}.
\end{lemma}
\begin{proof}
  For $x\in V$, set $\dist_x:= \dist_{G_s}(x,w)+1+\max_{y\in V}\{\delta_y\}$.
  Moreover, set $X_x:= \dist_x-\delta_x=\ell_{sx}+d_{G_s}(x,w)$ for $x\in V$ and let $X^{(i)}$ be the $i$'th order statistic of the variables $X_v$ (i.e., the $i$'th smallest element).
  Denote by $x_{\min}\in V$ the node for which $X_{x_{\min}}=X^{(1)}$.
  By Lemma~\ref{lemma:lemma4.4}, with constant probability $X^{(2)}-X^{(1)} \geq \frac{7}{8\beta}$. Condition on this event.
  Accordingly, we have for all $x\in V\setminus\{x_{\min}\}$ that $X_x-X_{x_{\min}}\geq X^{(2)}-X^{(1)} \geq \frac{7}{8\beta}$.

  Denote for each $v\in V$ by $x_v$ the child of $s$ in $T$ in whose subtree $v$ is situated.
  Then the assumption that $x_v\neq x_{\min}$ implies by copious use of the triangle inequality that
  \begin{align*}
  \dist_T(s,v)-\dist_{G_s}(s,v)&=\ell_{x_vs}+\dist_T(x_v,v)-\dist_{G_s}(s,v)\\
  &\geq \ell_{x_vs}+\dist_{G_s}(x_v,v)-\dist_{G_s}(s,v)\\
  &\geq \ell_{x_vs}+\dist_{G_s}(x_v,w)-\dist_{G_s}(v,w)-(\dist_{G_s}(s,w)+\dist_{G_s}(v,w))\\
  &\geq \ell_{x_vs}+\dist_{G_s}(x_v,w)-(\ell_{x_{\min}s}+\dist_{G_s}(x_{\min},w))-2\dist_{G_s}(v,w)\\
  &= X_{x_v}-X_{x_{\min}}-2\dist_{G_s}(v,w) \geq \frac{7}{8\beta}-2\dist_{G_s}(v,w).
  \end{align*}

  On the other hand, the approximation guarantee of the SSSP algorithm yields that
	\begin{equation*}
		\dist_T(s,v)-\dist_{G_s}(s,v)\leq \eps\dist_{G_s}(s,v)\leq \eps \ell_{vs}\leq \eps\max_{x\in V}\{1+\delta_x\}.
	\end{equation*}
	By Lemma~\ref{lemma:radius_bound}, w.h.p.\ $\max_{x\in V}\{\delta_x\}\leq t=\frac{c\log n}{4(1+\eps)\beta}$ after sampling the $\delta$-values in Line~\ref{line:shifts} of this iteration.
  Condition on this event as well.
  Using that $\varepsilon =\frac{1}{c\log^2 n}$ and $c$ is sufficiently large, we get that
  $\dist_T(s,v)-\dist_{G_s}(s,v) \leq \eps(1+t) < \frac{\beta}{4}$.

  In summary, if both events on which we conditioned occur, $x_v\neq x_{\min}$ entails that
  \begin{equation}\label{eq:ball}
  \dist_{G_s}(v,w)>\frac{5}{16\beta}.
  \end{equation}
  In particular, choosing $v=w$ yields the contradiction $0=\dist_{G_s}(w,w)>\frac{5\beta}{16}$, i.e., $x_w=x_{\min}$.

  We proceed to show that $\dist_G(v,w)\leq \frac{1}{40\beta}$ implies that also $v\in V_{x_{\min}}^{\circ}$.
  By a union bound over the two events on which we conditioned, this will complete the proof.
  To this end, observe that Inequality~\eqref{eq:ball} shows that a ball of radius $\frac{5}{16\beta}$ around $w$ in $G_s$ is contained within $V_{x_{\min}}$.
  Because longer edges have been deleted, nodes in $\partial \VVV$ are connected to neighbors outside their cell by edges of length at most $\frac{1}{40\beta}$.
  Together with the approximation guarantee of the second SSSP computation used to compute $T'$, it follows that nodes $v\in V$
  for which $\dist_{G_s}(v,w)\leq \frac{1/16-1/40-\varepsilon}{\beta}<
  \frac{5}{16\beta} - \frac{(1+\eps)^2}{4\beta}-\frac{1}{40\beta}$
  end up in $V_{x_{\min}}^{\circ}$.
  In particular, as trivially $\dist_{G_s}(v,w)\leq \dist_G(v,w)$ and
  $\varepsilon$ is sufficiently small, we conclude that $\dist_G(v,w)\leq
  \frac{1}{40\beta}$ implies that $v\in V_{x_{\min}}^{\circ}$.
\end{proof}
% From this lemma, it follows that Algorithm~\ref{alg:decomposition} terminates after $O(\log n)$ iterations w.h.p.
\begin{corollary}\label{cor:progress}
  Algorithm~\ref{alg:decomposition} terminates after $O(\log n)$ iterations of the while loop w.h.p.
\end{corollary}
\begin{proof}
  Consider any edge $e\in E$ that is not deleted right away, i.e., $\ell_e\leq \frac{1}{40\beta}$. By Lemma~\ref{lemma:progress}, in each iteration in which $e$ is present in the remaining subgraph of $G$, there is a constant probability that it is contained in $V_u^{\circ}$ for some node $u$. Thus, the probability that the edge remains for $c\log n$ iterations is bounded by $2^{-\Omega(c\log n)}=n^{-\Omega(c)}$.
  By a union bound, this implies that all edges are either cut or included in a part within $O(\log n)$ iterations w.h.p., i.e., the termination condition that $E(G)$ is empty becomes satisfied.
\end{proof}

%###
\paragraph{The Diameter Bound.}
%###
% We proceed to proving the claimed properties of the decomposition as stated in
% Theorem~\ref{theorem:main}.
In order to prove that the diameter bound holds, we first show that for each
iteration of the while loop of Algorithm~\ref{alg:decomposition} and each
$u\in C$, we have that $C_u\subseteq V_u$.
\begin{lemma}\label{lemma:subset}
  Fix any iteration of the while loop of Algorithm~\ref{alg:decomposition} and $u\in C$. It holds that $C_u\subseteq V_u$.
\end{lemma}
\begin{proof}
  Again, denote for simplicity the remaining subgraph at the beginning of the loop iteration by $G=(V,E)$.
  By the approximation guarantee of the second call to the SSSP algorithm, $v\in V_u^{\circ}$ implies that $\dist(v,\partial \VVV)\geq \frac{1}{4\beta}$.
  By Theorem~\ref{thm:blur}, $w\in C_u$ implies that $\dist_G(w,V_u^{\circ})\leq \frac{\rho}{1-1/(2\log n)}=\frac{1}{4\beta}$.
  Consider the node $v\in V_u^{\circ}$ that is closest to $w$ and fix a shortest path from $v$ to $w$.
  By the second bound, the path is no longer than $\frac{1}{4\beta}$, which by the first bound implies that it cannot leave $V_u$.
  Hence, $w\in V_u$, showing the claim of the lemma.
\end{proof}
We observe that the above lemma yields that the algorithm indeed outputs a partition of $V$, and each set in the partition is spanned by the corresponding tree in $\mathcal{T}$.
We now apply the tail bound on $\Exp_{\beta}$ given in Lemma~\ref{lemma:radius_bound} to infer that the diameter of the computed parts is bounded by $\frac{\Delta}{2}$ w.h.p.
\begin{lemma}\label{lemma:shrink}
	W.h.p., each connected component returned by {\normalfont\decomp$(G)$} has weak diameter at most $\frac{\Delta}{2}$. This is witnessed by the trees computed in Line~\ref{line:trees} of {\normalfont\decomp$(G)$}.
\end{lemma}
\begin{proof}
	By Lemma~\ref{lemma:radius_bound} and a union bound over all nodes, w.h.p.\ always $\max_{v\in V}\{\delta_v\}\leq 1+t$ for $t=\frac{c\log n}{4(1+\varepsilon)\beta}$ in Line~\ref{line:shifts} of \decomp$(G)$.
	Assume that $v$ ends up in the subtree of $T$ rooted at the child $x_v$ of $s$ in $G_s$. From the above bound, it follows that
	\begin{equation*}
		\dist_T(x_v,v)= \dist_T(s,v)-\ell_{x_vs}\leq
		\lfloor(1+\eps)\dist_{G_s}(s,v)\rfloor-1 \leq (1+\eps)\max_{u\in
		V}\{\delta_u\}<(1+\eps)t
	\end{equation*}
	w.h.p., where we exploited that edge weights are integral and that
	$\varepsilon<1$. Denoting by $C$ the children of the root node in $T$, it
	follows that for each $x\in C$, we have that $T_x$ has (weighted) depth at most
	$(1+\eps)t$ w.h.p.\ in Line~\ref{line:trees} of \decomp$(G, \Delta)$. We
	conclude that w.h.p., for all $u\in C$ it holds that the subgraph induced by
	$V_u$ has diameter at most
	$2(1+\eps)t=\frac{c\log n}{2\beta}=\frac{\Delta}{2}$.
  The claim of the lemma now follows immediately from Lemma~\ref{lemma:subset}.
\end{proof}

%###
\paragraph{The Edge Cutting Probability Bound.}
%###
We proceed to showing that the probability to cut an edge is sufficiently small. This follows from the analysis of Algorithm~\ref{alg:blur} and the probabilistic progress guarantee from Lemma~\ref{lemma:progress}.
\begin{corollary}\label{cor:cut}
	The probability that edge $e\in E$ is cut by \normalfont{\decomp}$(G, \Delta)$ is $O\left(\frac{\ell_e \log n}{\Delta}\right)$.
\end{corollary}
\begin{proof}
Consider edge $e=\{v,w\}\in E$. If $e$ is deleted right away, then $\ell_e>
\frac{1}{40\beta}= \Omega\left(\frac{\Delta}{\log n}\right)$ and the claim
trivially holds. Accordingly, assume that $\ell_e\leq \frac{1}{40\beta}$ in the
following.

As shown in Lemma~\ref{lemma:subset}, in each iteration the parts
$(V_u^\circ)_{u\in C}$ satisfy that $V_u^\circ\subseteq V_u$. Thus, if $v\in
V_x$ and $w\in V_y$ for some $x, y\in C$ after Line~\ref{line:trees}, $e$ can be
only cut by $v$ ending up in $C_x$, while $w$ does not, or $w$ ending up in
$C_y$, while $v$ does not. Lemma~\ref{lemma:prob_cut} shows that the probability
for either event is bounded by $O\left(\frac{\ell_e \log n}{\Delta}\right)$,
independently of the subgraph the calls to Algorithm~\ref{alg:blur} are executed
on.

Combining this observation with the fact that, in each iteration in which $e$
is still present by Lemma~\ref{lemma:progress}  it ends up in some part with
probability at least $p\in \Omega(1)$, we can bound the probability that $e$
is cut by
\begin{equation*}
\sum_{i=1}^{\infty} (1-p)^{i-1} O\left(\frac{\ell_e\log n}{\Delta}\right)
=O\left(\frac{\ell_e\log n}{p\Delta}\right)
=O\left(\frac{\ell_e\log n}{\Delta}\right).\qedhere
\end{equation*}
\end{proof}

%###
\paragraph{The Load Bound.}
%###
As the trees added to the output in a single iteration are subtrees of the same shortest path tree, these trees are disjoint.
Hence, the bound on the number of iterations also bounds the number of trees in which an edge may participate and thus the load of that edge in the output decomposition $D$.
This concludes the proof of Theorem~\ref{theorem:main}.
% Together, the previous statements prove Theorem~\ref{thm:decomposition}.
% \begin{proof}[Proof of Theorem~\ref{thm:decomposition}]
% The bound on the number of iterations of the algorithm is shown in Corollary~\ref{cor:progress}.
% By Lemma~\ref{lemma:subset}, it always holds that $C_u\subseteq V_u$.
% Accordingly, the algorithm indeed outputs a partition of $V$, and each set in the partition is spanned by the corresponding tree in $\mathcal{T}$.
% The w.h.p.\ bound on the depth of the trees is shown in Lemma~\ref{lemma:shrink}.
% Corollary~\ref{cor:cut} provides the bound on the probability to cut an edge.
% \end{proof}

%%% Local Variables:
%%% mode: latex
%%% TeX-master: "embeddable_tree"
%%% End:

%!TEX root=./embeddable_tree.tex
%%%%%%%%%%%%%%%%%%%%%%%%%%%%%%%%%%%%%%%%%%%%%%%%%%%%%%%%%%%%%%%%%%%%%%%%%%%%%%
\section{Sampling from Low Stretch Tree Embeddings}
\label{section:virtual-trees}
%%%%%%%%%%%%%%%%%%%%%%%%%%%%%%%%%%%%%%%%%%%%%%%%%%%%%%%%%%%%%%%%%%%%%%%%%%%%%%
Consider some graph
$G = (V, E, \ell)$.
We say that graph
$G' = (V', E', \ell')$
with
$V' \supseteq V$
\emph{dominates} $G$
if
$\dist_{G'}(u, v) \geq \dist_{G}(u, v)$
for every two vertices
$u, v \in V$.
In that case, we define the stretch of edge
$e=\{u,v\} \in E$
in $G'$ to be
\[
\str_{G'}(e)
\, = \,
\frac{\dist_{G'}(u,v)}{\ell_e} \, .
\]

Our goal in this section is to construct random dominating trees of a given
graph
$G = (V, E, \ell)$
that guarantee low expected stretch for each edge in $E$.
The dominating trees we construct, referred to hereafter as \emph{virtual
trees}, are not spanning trees of $G$, because they may include vertices and edges
that do not belong to $V$ and $E$, respectively.
Nevertheless, they admit some useful characteristics.
Specifically, we consider two types of virtual (dominating) trees:
\emph{projected trees} (a special case of the path embeddable trees of
\cite{CohenKMPPRX2014,CohenMPPX2014})
addressed in Section~\ref{section:projected-tree} and
\emph{hierarchically well separated trees (HSTs)} addressed in
Section~\ref{section:hst}.
In both cases, the respective constructions are based on recursive
applications of the graph decomposition technique presented in
Section~\ref{section:graph-decomposition}, generating a \emph{hierarchical}
version of TSDs as presented in
Section~\ref{section:hierarchical-decompositions}.

%%%%%%%%%%%%%%%%%%%%%%%%%%%%%%%%%%%%%%%
\subsection{Hierarchical Decompositions}
\label{section:hierarchical-decompositions}
%%%%%%%%%%%%%%%%%%%%%%%%%%%%%%%%%%%%%%%
A \emph{hierarchical tree-supported decomposition (HTSD)}
$\Seq{D}$
of $G$ is a sequence
$\Seq{D} = (D_{0}, D_{1}, \dots, D_{k})$
of TSDs that satisfies
(i)
$D_{0} = \{V\}$;
(ii)
$D_{k} = \{ \{ v \} \mid v \in V \}$;
and
(iii)
for every
$1 \leq i \leq k$
and
$C \in D_{i}$,
there exists some
$C' \in D_{i - 1}$
such that
$C \subseteq C'$.
The TSDs
$D_{0}, D_{1}, \dots, D_{k}$
are referred to as the \emph{levels} of $\Seq{D}$ and the parameter $k$ is
referred to as its \emph{depth}.
The \emph{load} of edge
$e \in E$
in $\Seq{D}$ is defined to be
$\Load_{\Seq{D}}(e) = \sum_{i = 0}^{k} \Load_{D_{i}}(e)$.

The real sequence
$\Seq{d} = (d_{0}, d_{1}, \dots, d_{k})$
is said to be \emph{diameter bounding} for the HTSD $\Seq{D}$ if
$\diam(D_{i}) \leq d_{i}$
for every
$0 \leq i \leq k$.
Of particular interest are HTSDs that admit a \emph{geometrically decreasing}
diameter bounding sequence, namely a sequence
$\Seq{d} = (d_{0}, d_{1}, \dots, d_{k})$
that satisfies
$d_{i} \leq \alpha \cdot d_{i - 1}$,
$1 \leq i \leq k$,
for some constant
$\alpha > 1$.
Since all edge lengths considered in this paper are integers bounded by some
polynomial in $n$, this means that $\Seq{D}$ admits
$k = O (\log n)$
levels.

Consider some HTSD
$\Seq{D} = (D_{0}, D_{1}, \dots, D_{k})$
of $G$ with a geometrically decreasing diameter bounding sequence
$\Seq{d} = (d_{0}, d_{1}, \dots, d_{k})$.
Edge
$e = \{ u, v \} \in E$
is said to be \emph{decoupled} on level
$0 \leq i \leq k - 1$
if $u$ and $v$ belong to the same cluster in level $i$ and to different
clusters in level
$i + 1$,
that is
$e \in \Ecut(D_{i + 1}) - \Ecut(D_{i})$.
In that case, we define the \emph{stretch} of $e$ in $\Seq{D}$ with respect to
$\Seq{d}$ to be
\[
\str_{\Seq{D}, \Seq{d}}(e)
\, = \,
\frac{d_{i}}{\ell_e} \, .
\]

In order to construct the HTSD, we first compute a
$4$-approximation $\Delta$ of the diameter $\diam(G)$.
For this purpose, we pick an arbitrary node $s\in V$ and compute
a $2$-approximate SSSP tree $T$ with source $s$.
We then let $\Delta:=\frac{\max_{x\in V}\{\dist_T(s, x)\}}{2}$.
\begin{observation}\label{obs: Delta bound}
  $\Delta\in[\frac{\diam(G)}{4}, \diam(G)]$.
\end{observation}
\begin{proof}
  It holds that
  \[
    \Delta
    % = \frac{\tilde{r}}{1+\eps}
    =\frac{\max_{x\in V}\{\dist_T(s, x)\}}{2}
    \leq \max_{x\in V}\{\dist(s, x)\}
  \leq \diam(G).
  \]
  For the lower bound, note that
  \[
    \Delta
    % = \max_{x\in V}\left\{\frac{\tilde{\dist}(v)}{1+\eps}\right\}
    % =\frac{\max_{x\in V}\{\tilde d(x)\}}{1+\eps}
    \geq \frac{\max_{x\in V}\{\dist(s, x)\}}{2}
    \geq \frac{\max_{x,y\in V}\{\dist(s, x)+\dist(s, y)\}}{4}\geq \frac{\diam(G))}{4}.\qedhere
  \]
\end{proof}

We proceed to showing how to construct a random HTSD.
\begin{theorem} \label{theorem:hierarchical-tsd}
There exists a (randomized) algorithm that, given a graph
$G = (V, E, \ell)$
with $\poly(n)$-bounded edge lengths, constructs a random HTSD
$\Seq{D}$ of $G$ with the
following guarantees:
(1)
the depth of $\Seq{D}$ is
$O (\log n)$;
(2)
$\Seq{D}$ admits a (deterministic) geometrically decreasing diameter bounding
sequence
$\Seq{d}$
w.h.p.;
(3)
$\Load_{\Seq{D}}(e) = O (\log n)$
for every edge
$e \in E$
w.h.p.;
and
(4)
$\Ex_{\Seq{D}} [ \str_{\Seq{D}, \Seq{d}}(e) ] = O (\log^{2} n)$
for every edge
$e \in E$.
\end{theorem}
\begin{proof}
Let
$d_{0} = 4 \Delta$, with $\Delta$ being a constant factor approximation of $\diam(G)$ as above,
and let
$d_{i} = \frac{d_{i - 1}}{2}$
for
$i \geq 1$.
Let $k$ be the smallest $i$ such that
$d_{k} < 1$.
Since the edge lengths in $G$ are $\poly(n)$-bounded, we know that
$k = O (\log n)$.

We construct the (random) HTSD
$\Seq{D} = (D_{0}, D_{1}, \dots, D_{k})$
of $G$ by applying \decomp (see
Section~\ref{section:graph-decomposition}) in a recursive manner with diameter
bounds determined according to the sequence
$\Seq{d} = (d_{0}, d_{1}, \dots, d_{k})$.
Corollary~\ref{cor:cut} guarantees that for every edge
$e \in E$
and level
$0 \leq i \leq k - 1$,
the probability that $e$ is decoupled on level $i$ of a random HTSD sampled
from $\mathcal{S}$ is in
$O \big( \frac{\ell_e \cdot \log (n)}{d_{i}} \big)$.
The bound on
$\Ex_{\Seq{D}} [\str_{\Seq{D}, \Seq{d}}(e)]$
follows directly by summing over all levels
$0 \leq i \leq k - 1$.

It remains to show that the load of every edge
$e \in E$
in $\Seq{D}$ is
$O (\log n)$
w.h.p.
To that end, recall that in Section~\ref{section:graph-decomposition} we
proved that the load on edge $e$ in the TSD $D_{i}$ is stochastically
dominated by a geometric random variable with parameter
$\Omega (1)$.
The claim follows, as the sum of
$O (\log n)$
such random variables is
$O (\log n)$
w.h.p.
\end{proof}
We note that a crucial point is, of course, that the algorithm can be
implemented efficiently due to relying on approximate SSSP computations only.
However, as the resulting complexities are model-specific, the respective
discussion is postponed to Section~\ref{sec:implementation}.

%%%%%%%%%%%%%%%%%%%%%%%%%%%%%%%%%%%%%%%
\subsection{Embedding into a Random Projected Tree}
\label{section:projected-tree}
%%%%%%%%%%%%%%%%%%%%%%%%%%%%%%%%%%%%%%%
Consider some graph
$G = (V, E, \ell)$.
Graph
$G' = (V', E', \ell')$
with
$V' \supseteq V$
is said to be a \emph{projected} graph of $G$ if there exists a mapping
$\pi : V' \rightarrow V$
so that
\begin{compactenum}[(a)]
  \item $\pi(v) = v$
  for every
  $v \in V$;
  \item if
  $e' = \{ u', v' \} \in E'$,
  then
  $\pi(e') := \{ \pi(u'), \pi(v') \} \in E$; and
  \item $\ell'(e') = \ell(e)$
  for every
  $e \in E$
  and
  $e' \in E'$
  such that
  $\pi(e') = e$.
\end{compactenum}
The \emph{load} of edge
$e \in E$
under the projected graph $G'$ of $G$ (with respect to $\pi$) is defined to be
the size of its preimage under $\pi$, denoted by
$\Load_{G'}(e) = |\{ e' \in E' \mid \pi(e') = e \}|$.
Notice that, by definition, every projected graph of $G$ dominates $G$.
Observe also that $\ell'$ is fully determined by $\pi$ and $\ell$, hence we may
omit it from the notation in the following.
Our goal in this section is to prove the following theorem.

\begin{theorem} \label{theorem:projected-tree}
There exists a (randomized) algorithm that, given a graph
$G = (V, E, \ell)$
with $\poly(n)$-bounded edge lengths, constructs a random projected
tree $T$ of $G$ that satisfies the following guarantees for every edge
$e \in E$:
(1)
$\Load_{T}(e) = O (\log n)$
w.h.p.;
and
(2)
$\Ex_{T} [\str_{T}(e)] = O (\log^{2} n)$.
\end{theorem}

Theorem~\ref{theorem:projected-tree} is established by combining
Theorem~\ref{theorem:hierarchical-tsd} with the following lemma.

\begin{lemma} \label{lemma:htsd-to-projected-tree}
There exists an algorithm that given a graph
$G = (V, E, \ell)$,
a HTSD $\Seq{D}$ of $G$, and a geometrically decreasing diameter bounding
sequence $\Seq{d}$ for $\Seq{D}$, constructs a projected tree
$T = (V_{T}, E_{T}, \ell_{T})$
of $G$ such that
$\Load_{T}(e) = \Load_{\Seq{D}}(e)$
and
$\str_{T}(e) = O (\str_{\Seq{D}, \Seq{d}}(e))$
for each $e \in E$.
\end{lemma}

The rest of Section~\ref{section:projected-tree} is dedicated to proving
Lemma~\ref{lemma:htsd-to-projected-tree}.
This is done by a series of graph transformations that results in the desired
projected tree $T$.
Let $k$ be the depth of
$\Seq{D} = (D_{0}, D_{1}, \dots, D_{k})$.
For
$0 \leq i \leq k$,
let
$H_{i} = (V^{H}_{i}, E^{H}_{i})$
be the forest obtained by taking the (graph) union over all level $i$
supporting trees of $\Seq{D}$, where each level $i$ supporting tree
$T_{C} = (U_{C}, F_{C})$
contributes its own (distinct) copies of the vertices in $U_{C}$ and edges in
$F_{C}$ (this means, in particular, that
$|V^{H}_{i}| = \sum_{C \in D_{i}} |U_{C}|$
and
$|E^{H}_{i}| = \sum_{C \in D_{i}} |F_{C}|$).
Define the function
$\pi^{H}_{i} : V^{H}_{i} \rightarrow V$
by mapping each vertex
$v \in V^{H}_{i}$
to the vertex
$\pi^{H}_{i}(v) \in V$
from which it originates, recalling that $T_{C}$ is a subgraph of $G$.
Although the preimage of vertex
$u \in V$
under $\pi^{H}_{i}$ may consist of several vertices, it includes exactly one
vertex
$v_{i} \in U_{C}$,
where $C$ is the (unique) level $i$ cluster that contains $v$.
We hereafter refer to this vertex $v_{i}$ as the level $i$ \emph{clone}
of $v$.

Recalling that the level $k$ clusters of $\Seq{D}$ are singletons, we
identify the vertices in $V^{H}_{k}$ with their images under (the bijection)
$\pi^{H}_{k}$ so that
$V^{H}_{k} = V$.
Let
$H = (V^{H}, E^{H})$
be the forest obtained by taking the (graph) union over
$H_{0}, H_{1}, \dots, H_{k}$
and let
$\pi^{H} \colon V^{H} \to V$
be the function defined by mapping each vertex
$v \in V^{H}_{i}$,
$0 \leq i \leq k$,
to
$\pi^{H}(v) = \pi^{H}_{i}(v)$.
Notice that $H$ is a projected graph of $G$ realized by $\pi^{H}$ and that
$\Load_{H}(e) = \Load_{\Seq{D}}(e)$
for every edge
$e \in E$.
It remains to show that we can turn $H$ into a projected tree
$T = (V_{T}, E_{T})$
by connecting its connected components without increasing the load on the
edges while ensuring that the stretch of every edge
$e \in E$
in $T$ is at most
$O (1)$
times larger than its stretch in $\Seq{D}$ with respect to $\Seq{d}$.

Given a level
$0 \leq i \leq k$
and a level $i$ cluster $C$, we refer to the vertex with smallest ID in $C$ as
the \emph{leader} of cluster $C$, denoted by
$\lambda(C)$.
Notice that every vertex
$v \in V$
is a leader of its level $k$ cluster and that if $v$ is the leader of its
level $i$ cluster, then it is also the leader of its level $j$ cluster for
all
$i \leq j \leq k$.

We now construct a projected tree
$T = (V^{T}, E^{T})$
of $G$ from $H$ in two additional steps.
First, we connect each connected component $T_{C}$ of $H_{i}$,
$1 \leq i \leq k$,
to the unique connected component $T_{C'}$ of
$H_{i - 1}$
that satisfies
$C' \supseteq C$.
Assuming that the leader of cluster $C$ is
$v = \lambda(C)$,
this connection is realized by augmenting $H$ with a $0$-length edge that
connects $v_{i}$ with
$v_{i - 1}$,
i.e., the level $i$ and level $i - 1$ clones of $v$.
(Note that
$v_{i - 1}$
is not necessarily the leader of cluster $C'$.)
We call this new edge connecting $v_{i}$ and
$v_{i - 1}$
a \emph{vertical} edge and denote the set of all vertical edges added to $H$
during this step of the construction by $E_{\uparrow}$.
Observe that the graph obtained from $H$ by augmenting it with the vertical
edges is a tree denoted hereafter by
$T^{\uparrow}=(V^{H}, E^{H} \cup E_{\uparrow})$.
This holds since starting from the forest $H$, we connected each connected
component in level
$1 \leq i \leq k$
to a connected component in level
$i - 1$
using a single vertical edge and since $H_{0}$ is a tree.

The next and final step simply contracts all vertical edges in $T^{\uparrow}$,
resulting in the tree
$T = (V^{T}, E^{T})$.
Since the vertical edge
$\{ v_{i}, v_{i - 1} \} \in E^{\uparrow}$
connects the clones $v_{i}$ and
$v_{i - 1}$
of the same vertex
$v \in V$,
it follows that both endpoints of the vertical edge are mapped to $v$ under
$\pi^{H}$.
Accordingly, we readily obtain a projection
$\pi^{T} \colon V^{T} \to V$
from $\pi^{H}$ by mapping each vertex
$v^{T} \in V^{T}$
to $\pi^{H}(v')$, where
$v' \in V^{H}$
is any node that participated in the contraction that created $v^{T}$.
Finally, note that there is a natural bijection
$b \colon E^{H} \to E^{T}$
between edges in $H$ and $T$, as $T$ is obtained by first augmenting $H$ with
the set $E^{\uparrow}$ of vertical edges and then contracting these edges.
By construction, we have that
$\pi^{T}(b(e)) = \pi^{H}(e)$
for all
$e \in E^{H}$.
In particular, $T$ is indeed a projected tree of $G$ and
$\Load_{T}(e) = \Load_{H}(e) = \Load_{\Seq{D}}(e)$
for all
$e \in E$.

It remains to prove that
$\str_{T}(e) = O (\str_{\Seq{D},\Seq{d}}(e))$
for every edge
$e=\{x,y\} \in E$.
Since $T$ is obtained from $T^{\uparrow}$ by contracting $0$-length edges, it
follows that
$\dist_{T^{\uparrow}}(x, y) = \dist_{T}(x, y)$,
hence it suffices to prove that
$\str_{T^{\uparrow}}(e) = O (\str_{\Seq{D},\Seq{d}}(e))$.
To this end, fix some node
$v \in V$
and let
$C_{i} \in D_{i}$,
$0 \leq i \leq k$,
be the (unique) level $i$ cluster that contains $v$.
Let
$\lambda(i) = \lambda(C_{i})$
be the leader of $C_{i}$ and denote the level $j$ clone of $\lambda(i)$ by
$\lambda_{j}(i)$.

\begin{observation} \label{obs:stretch}
For every
$0 \leq i \leq k$,
we have
$\dist_{T^{\uparrow}}(v, \lambda_{i}(i)) \leq \sum_{j = i}^{k - 1} d_{j}$.
\end{observation}
\begin{proof}
By induction on $i$.
The base case
$i = k$
holds since every vertex is the leader of its (singleton) level $k$ cluster,
hence
$\lambda_{i}(i) = v$.
For the inductive step from
$i + 1$
to
$0 \leq i \leq k - 1$,
we notice that
\[
\dist_{T^{\uparrow}} \left( v, \lambda_{i}(i) \right)
\, = \,
\dist_{T^{\uparrow}} \left( v, \lambda_{i + 1}(i + 1) \right)
+
\dist_{T^{\uparrow}} \left( \lambda_{i + 1}(i + 1), \lambda_{i}(i + 1) \right)
+
\dist_{T^{\uparrow}} \left( \lambda_{i}(i + 1), \lambda_{i}(i) \right) \, .
\]
Recalling that
$\lambda_{i + 1}(i + 1)$
and
$\lambda_{i}(i + 1)$
are connected in $T^{\uparrow}$ by a vertical edge, we conclude that
$\dist_{T^{\uparrow}} \left( \lambda_{i + 1}(i + 1), \lambda_{i}(i + 1) \right)
=
0$.
Moreover, since
$\lambda_{i}(i + 1)$
and
$\lambda_{i}(i)$
belong to the same level $i$ cluster
$C_{i} \in D_{i}$,
their distance in $T^{\uparrow}$ is equal to their distance in the
supporting tree of $C_{i}$ whose diameter is bounded by $d_{i}$,
hence
$\dist_{T^{\uparrow}} \left( \lambda_{i}(i + 1), \lambda_{i}(i) \right)
\leq
d_{i}$.
The assertion follows by the inductive hypothesis ensuring that
$\dist_{T^{\uparrow}} \left( v, \lambda_{i + 1}(i + 1) \right)
\leq
\sum_{j = i + 1}^{k - 1} d_{j}$.
\end{proof}

Now, consider some edge
$e = \{u, v\} \in E$
and let
$0 \leq i \leq k - 1$
be the level on which $e$ is decoupled.
Let
$C \in D_{i}$
to be the level $i$ cluster that contains $u$ and $v$ and let $w$ be the level
$i$ clone of the leader $\lambda(C)$ of $C$.
Observation~\ref{obs:stretch} guarantees that
$\dist_{T^{\uparrow}}(u, w) \leq \sum_{j = i}^{k - 1} d_{j}$
and
$\dist_{T^{\uparrow}}(v, w) \leq \sum_{j = i}^{k - 1} d_{j}$,
hence
\[
\dist_{T^{\uparrow}}(u, v)
\, \leq \,
2 \sum_{j = i}^{k - 1} d_{j}
\, = \,
O (d_{i}) \, ,
\]
where the last transition holds since
$\Seq{d} = (d_{0}, d_{1}, \dots, d_{k})$
is geometrically decreasing.
The proof of Lemma~\ref{lemma:htsd-to-projected-tree} is completed by the
definitions of
$\str_{\Seq{D}, \Seq{d}}(e) = \frac{d_{i}}{\ell_{e}}$
and
$\str_{T^{\uparrow}}(e) = \frac{\dist_{T^{\uparrow}}(u, v)}{\ell_{e}}$.

%%%%%%%%%%%%%%%%%%%%%%%%%%%%%%%%%%%%%%%
\subsection{Embedding into a Random HST}
\label{section:hst}
%%%%%%%%%%%%%%%%%%%%%%%%%%%%%%%%%%%%%%%
In this section we show how to construct an embedding
into a random \emph{hierarchically $2$-separated} dominating tree (HST)
with small expected stretch from the projected trees
constructed in the previous section.
\begin{definition}[Hierarchically Separated Trees]\label{definition: hst}
  An \emph{embedding} of a weighted graph $G=(V,E,\ell)$
  into a (rooted) tree $T=(V^T,E^T,\ell^T)$ is given by a
  one-to-one mapping $\iota\colon V\to V^T$.
  For $k>1$, the tree is \emph{hierarchically $k$-separated},
  if for each internal non-root node,
  the weight of edges connecting it to its children is
  exactly by factor $k$ smaller
  than the weight of the edge connecting it to its parent.
  The stretch of edge $e=\{u,v\}\in E$ w.r.t.\ $T$ is defined as
  $\str_T(e):=\frac{\dist_T(\iota(u),\iota(v))}{\ell_e}$.
\end{definition}
We note that our definition of hierarchical well-separation is
(formally) weaker than that of hierarchically well-separated trees
from the literature~\cite{Bartal1996},
as we dropped the requirement that the tree is balanced,
i.e., all leaves are in the same depth.
However, this can be easily achieved, and our construction does so without
modification.

%###
\paragraph{Construction.}
%###
We construct our HST from a projected tree (see
Section~\ref{section:projected-tree}).
The construction of $T=(V^T, E^T, \ell^T)$ is straightforward.
Let $\Seq{D}=(D_0,\ldots, D_k)$ be the HTSD from which the projected tree
was constructed.
We recall that we had assigned a leader $\lambda(C)$ to each cluster $C$,
namely the smallest ID vertex in $C$.
We construct $V^T$ simply as the multiset\footnote{For each cluster $C$ a node
$v\in V$ is leader of, there is a separate copy of $v$.} of leaders of all
clusters in $\Seq{D}$.
Note that the nodes constructed for level $k$ clusters, correspond,
one-to-one, to the original nodes $V$ of the graph.
This enables us to define an embedding $\iota:V\rightarrow V^T$ as required
in Definition~\ref{definition: hst}.
We construct the set of edges $E^T$ as follows: Let $\lambda\in V^T$ be a node
corresponding to an arbitrary level $i$ cluster $C$ with $i<k$.
We introduce an edge $e:=\{\lambda(C), \lambda(C')\}$,
for every level $i+1$ cluster $C'$
that cluster $C$ decomposes into, i.e., $C'\subseteq C$.
We assign length $\ell^T_e:=d_i$ to
such an edge $e$ between
nodes corresponding to level $i$ and level $i+1$ clusters.
Rooting the tree at the node in $V^T$ corresponding to the leader of the
(unique) level $0$ cluster $V$, it is clear that the resulting tree
$T=(V^T, E^T, \ell^T)$ is a hierarchically $2$-separated tree of depth
$O(\log n)$ w.h.p.

Regarding distances, we get essentially the same result as for the projected
tree we could have constructed. Denote for $v\in V$ by $\lambda(i)$ the leader
of the unique level $i$ cluster $C_i\in D_i$ such that $v\in C_i$ and denote by
$\lambda^T(i)\in V^T$ its copy in $T$ corresponding to $C_i$.
\begin{observation}\label{obs:stretch_HST}
For every
$0 \leq i \leq k$,
we have
$\dist_{T}(v, \lambda^T(i)) = \sum_{j = i}^{k - 1} d_{j}$.
\end{observation}
\begin{proof}
$\dist_{T}(v,\lambda^T(i))
=\sum_{j=i}^{k-1}\ell^T_{\{\lambda^T(j),\lambda^T(j+1)\}}
=\sum_{j=i}^{k-1}d_j$.
\end{proof}
\begin{corollary}
$T$ is a dominating hierarchically $2$-separated tree with
$\E_T[\str_T(e)]=O(\log^2 n)$ for each edge $e\in E$.
\end{corollary}
\begin{proof}
As discussed, $T$ is hierarchically $2$-separated by construction and we have
the desired embedding $\iota\colon V\to V^T$.
By Observations~\ref{obs:stretch} and~\ref{obs:stretch_HST}, distances
between leaves of $T$ are at least as large as in the projected tree constructed
in Section~\ref{section:projected-tree}, which dominates $G$. The stretch bound
follows analogous to Section~\ref{section:projected-tree}, where
Observation~\ref{obs:stretch} takes the place of Observation~\ref{obs:stretch_HST}.
\end{proof}
We remark that this establishes a straightforward relation between our projected
trees and the HSTs constructed here. The HST edges are realized by the
corresponding paths in the projected tree. In particular, while the HST may
incur large loads on some graph edges, the ``more fine-grained'' view provided
by the projected tree shows that a low-load mapping of paths in the HST to the
original graph is feasible. On the other hand, this relation also demonstrates
that a projected tree ``behaves'' like an HST due to the geometrically
decreasing diameter bounding sequence of the underlying HTSD.

%%%%%%%%%%%%%%%%%%%%%%%%%%%%%%%%%%%%%%%
\subsection{Bounding the \texorpdfstring{$p$}{p}-Stretch}
\label{section:p-stretch}
%%%%%%%%%%%%%%%%%%%%%%%%%%%%%%%%%%%%%%%
Cohen et al.~\cite{CohenMPPX2014} introduced the notion of $p$-stretch.
\begin{definition}[$p$-Stretch]
For a graph $G$, an embedding of $G$ into $T$, and a real $p\in (0,1]$, the
\emph{$p$-stretch} of an edge $e=\{u,v\}\in E$ is given by
$\left(\frac{\dist_T(u,v)}{\ell_e}\right)^p$. Analogously, we define the
$p$-stretch of an HTSD for edge $e$ as $\left(\frac{d_i}{\ell_e}\right)^p$,
where $i$ is the level on which $e$ is decoupled.
\end{definition}
Note that the $1$-stretch coincides with the definition of the standard stretch
defined at the beginning of this section. Our constructions meet a stronger
bound of $O(\log n)$ on the $p$-stretch for $p<1$, owed to the fact that for
$p<1$ larger stretch is weighed less.

\begin{lemma}\label{lemma:p_stretch}
For $p\in (0,1)$, the tree embeddings presented in
Sections~\ref{section:projected-tree} and~\ref{section:hst} satisfy that for
each edge $e\in E$ the expected $p$-stretch is $O(\log n)$.
\end{lemma}
\begin{proof}
When bounding the stretch in the proof of
Theorem~\ref{theorem:hierarchical-tsd}, we summed over all levels of the
decomposition. Recall that the probability to decouple edge $e$ on level $i$ is,
by Corollary~\ref{cor:cut}, $O\left(\frac{\ell_e\cdot \log n}{d_i}\right)$.
Denote by $i_e$ the level such that $d_{i_e}\leq \ell_e<2d_{i_e}$. If $i>i_e$,
then the stretch of $e$ w.r.t.\ the HTSD is smaller than $1$. For $p<1$, the sum
now can thus be bounded as
\begin{align*}
\sum_{i=1}^k O\left(\frac{\ell_e\cdot \log
n}{d_i}\right)\cdot\left(\frac{d_i}{\ell_e}\right)^p
&=O\left(1+\log n\cdot
\left(\frac{\ell_e}{d_{i_e}}\right)^{1-p}\cdot\sum_{i=1}^{i_e}
\left(\frac{d_{i_e}}{d_i}\right)^{1-p}\right)\\
&=O\left(1+\log
n\cdot\sum_{i=1}^{i_e}\left(\frac{1}{2}\right)^{(1-p)(i_e-i)}\right)\\
&=O(\log n),
\end{align*}
where the final step exploits that the sum is a geometric series due to $1-p>0$.
\end{proof}

%%% Local Variables:
%%% mode: latex
%%% TeX-master: "embeddable_tree"
%%% End:

%!TEX root = ./embeddable_tree.tex
\section{Implementation in Different Models}\label{sec:implementation}
In this section, we describe how to implement the above techniques
in the \Congest, \Pram, and multipass streaming models.
These should be considered as exemplary computational models
and it seems likely that our techniques transfer to other models in which a
discrepancy between exact and approximate SSSP computations exist.
For the \Congest model some effort is needed in order to transfer the
$(1+\eps)$-approximate SSSP result
to $(1+\eps)$-approximate SSSP in super-source graphs
(see Definition~\ref{def:super}),
while for the other two models this is immediate.

\subsection{\texorpdfstring{\Congest}{CONGEST} Model}
\label{subsec:congest}
In the \Congest model of computation~\cite{Peleg2000}, every node is a computing
unit (of unlimited computational power) and is labeled by a unique $O(\log
n)$-bit identifier. Computation proceeds in synchronous rounds, in each of which
a node (1) performs local computations, (2) sends $O(\log n)$-bit messages to
its neighbors, and (3) receives the messages that its neighbors sent.
Initially, every node in the input graph $G=(V, E,\ell)$ knows its identifier
and its incident edges together with their length. We note that the restriction
to polynomially bounded edge lengths implies that distances can be encoded
using $O(\log n)$ bits.

At termination every node needs to know its part of the output.
For the task of constructing the random TSD, this means that every node
$v \in V$
knows
(1)
the ID of its own cluster's leader (i.e., the vertex with minimum ID, see
Section~\ref{section:projected-tree});
(2)
the ID of the leader of cluster $C$ if
$v \in U_{C}$,
that is, if $v$ participates in the supporting tree $T_{C}$ of cluster $C$;
and
(3)
its incident edges in $T_{C}$ for each supporting tree $T_{C}$ in which $v$
participates.
For the task of constructing the random HTSD, $v$ should hold that knowledge
for every level of the hierarchy.
As discussed in Section~\ref{section:virtual-trees}, this also provides the
nodes with all what they need in order to reconstruct the resulting projected
tree or HST.

In order to avoid confusion with the weighted diameter $\diam(G)$,
in what follows, we use $\hop(G)$ to denote the
``unweighted'' diameter of $G$, also called the \emph{hop diameter}.

The following corollary discusses how to compute $(1+\eps)$-approximate SSSP
in a super-source graph $H$ of a graph $G$ in the \Congest model.
We assume that each node $v\in V$ initially knows which of its incident
edges in $G$ are in $H$, whether it is connected to $s$, and, if so, the length
$\ell(\{s,v\})$.
\begin{corollary}[of~\cite{DBLP:conf/wdag/BeckerKKL17}]\label{cor:bkkl}
  Let
  $\varepsilon= \frac{1}{\polylog n}$.
  Then $(1+\varepsilon)$-approximate SSSP in super-source graphs can be solved
  in $\tilde{O}(\sqrt{n}+\hop(G))$ rounds w.h.p. in the \Congest model.
\end{corollary}
\begin{proof}
  The algorithm from~\cite{DBLP:conf/wdag/BeckerKKL17} consists of three main
  steps:
  \begin{enumerate}
  \item Let $S$ be a set composed of $s$ and $\tilde{\Theta}(\sqrt{n})$ nodes
  sampled uniformly at random. Let each node $v\in S$ learn a
  $(1+\frac{\varepsilon}{3})$-approximation to the minimum length of
  $\tilde{O}(\sqrt{n})$-hop paths to each sampled node $w\in S$ (if no such node
  exists, any result of at least $d(v,w)$ is fine, including $\infty$). For each
  finite value, nodes on a (unique) path in $G$ learn about them being part of
  this path and the next node on it.
  \item Simulate a broadcast congested clique\footnote{The broadcast congested
  clique is
  the special case of the Congest model restricted to complete graphs and, for
  each round, nodes sending the same message to each of their neighbors.}
  $(1+\varepsilon/3)$-approximate SSSP algorithm on the (virtual) graph on $S$
  with edge lengths given by the result from the previous step
  ($\infty$ means no edge).
  \item Run $\tilde{O}(\sqrt{n})$ iterations of single source Bellman-Ford on
  $G$, where the distance values of nodes in $S$ are initialized to the
  distances obtained from the previous step.
  \end{enumerate}
  Assuming w.l.o.g.\ that $\varepsilon\leq 1$, this yields
  $(1+\varepsilon)$-approximate distances to $s$. As the first step yields
  suitable routing information and the result of the second (i.e., an
  approximate SSSP-tree on the virtual graph) is global knowledge,
  nodes can locally determine their parent in the output tree~$T$.

  We adapt the algorithm to super-source graphs as follows.
  \begin{enumerate}
    \item The first step is based on a pipelined version of the (multi-source)
    Bellman-Ford algorithm that also works on directed
    graphs~\cite[Corollary~5.8]{LenzenPP2018}. Formally, we orient all edges of
    $s$ towards it (no other change is made). Then we can easily simulate the
    procedure on the resulting graph, as all communication by $s$ over one of
    its edges can be inferred from its length (which is known to the recipient).

    Note that the result is not exactly the same as that of the first step
    above: all paths containing $s$ as non-starting node have been removed.
    However, the decisive property of the constructed graph is that it
    preserves $G$-distances to $s$ up to a factor of $1+\varepsilon$.
    The virtual graph also needs to be undirected, which is achieved by
    dropping the directionality of the computed distances.
    \item The simulation of the broadcast congested clique algorithm in the
    Congest model is based on making all communication global knowledge.
    Using pipelining over a BFS tree, the input of $s$ in the virtual graph
    (i.e., its incident edges and their lengths) can be made global knowledge
    in $\tilde{O}(\sqrt{n}+\hop(G))$ rounds. Together, this implies that all
    nodes can locally simulate $s$.
    \item Simulating the communication by $s$ in the third step of the
    algorithm, which is a standard Bellman-Ford computation, is straightforward.
  \end{enumerate}
  As all steps can be adjusted preserving the guarantees of the algorithm and
  the asymptotic running time is increased by additive
  $\tilde{O}(\sqrt{n}+\hop(G))$ only,
  the result now follows from~\cite{DBLP:conf/wdag/BeckerKKL17}.
\end{proof}

This leads to the following result for Algorithm~\ref{alg:blur} from
Section~\ref{sec: blur}. As it is basically a sequence of approximate
SSSP computations, a running time bound is immediate from
Corollary~\ref{cor:bkkl}.
\begin{corollary}\label{cor:blur_congest}
  Suppose $\alpha= \frac{1}{\polylog n}$ and $\rho= n^{O(1)}$. Then
  Algorithm~\ref{alg:blur} can be executed in the \Congest model in
  $\tilde{O}(\sqrt{n}+\hop(G))$ rounds w.h.p.
\end{corollary}
\begin{proof}
The while loop terminates after at most $\lceil\log_{1/\alpha}\rho = O(\log
n)\rceil$ iterations. In each iteration, $r^{[i]}$ can be chosen by an arbitrary
node (e.g.\ the one with lowest identifier) and broadcasted via a BFS tree in
$O(\hop(G))$ rounds. Each node then can infer from the result from the previous
iteration (or the input if $i=1$) whether it is part of $B^{[i-1]}$. Nodes
adjacent to $B^{[i-1]}$ can learn about this in one communication round and
infer the length of the edge connecting them to $s^{[i]}$ in $G^{[i]}$. Thus,
all that remains is the approximate SSSP computation, which can be performed in
the stated running time by Corollary~\ref{cor:bkkl} w.h.p. The
$\tilde{O}$-notation absorbs the $O(\log n)$-factor from the number of loop
iterations.
\end{proof}

We turn to Algorithm~\ref{alg:decomposition} from
Section~\ref{section:graph-decomposition}. As each iteration can be performed
within $\tilde{O}(\sqrt{n}+\hop(G))$ rounds w.h.p., this implies a bound on the
running time of the overall algorithm.
\begin{corollary}\label{cor:decomposition_congest}
If $\Delta= n^{O(1)}$, Algorithm~\ref{alg:decomposition} can be executed
within $\tilde{O}(\sqrt{n}+\hop(G))$ rounds w.h.p.
\end{corollary}
\begin{proof}
All computations with the exception of the approximate SSSP computations and the
call to \texttt{blur} are local. By Corollary~\ref{cor:blur_congest}, the stated
running time bound follows for a single iteration of the while loop. Here we use
that the instances of \texttt{blur} can be run in parallel by
Lemma~\ref{lemma:subset}: As $C_u\subseteq V_u$, we can delete all edges which
are not connecting two nodes within the same $V_u$ for some $u$ and then run a
single $(1+\eps)$-approximate SSSP instance, where we identify the super-sources
of all calls to \texttt{blur}. Therefore, Corollary~\ref{cor:progress} and a
union time bound yield the claim.
\end{proof}

We now turn to the techniques from Section~\ref{section:virtual-trees}.
As the recursive calls for each level of the decomposition hierarchy when
computing an HTSD can be executed concurrently with a single call to the
approximate SSSP subroutine, we obtain the following result.
\begin{corollary}\label{cor:htsd_congest}
There exists a \Congest algorithm that, given a graph
$G = (V, E, \ell)$
with $\poly(n)$-bounded edge lengths, constructs a random HTSD
$\Seq{D}$ of $G$ with the
following guarantees in $\tilde{O}(\sqrt{n}+\hop(G))$ rounds w.h.p.:
(1)
the depth of $\Seq{D}$ is
$O (\log n)$;
(2)
$\Seq{D}$ admits a (deterministic) geometrically decreasing diameter bounding
sequence
$\Seq{d}$
w.h.p.;
(3)
$\Load_{\Seq{D}}(e) = O (\log n)$
for every edge
$e \in E$
w.h.p.;
and
(4)
$\Ex_{\Seq{D}} [ \str_{\Seq{D}, \Seq{d}}(e) ] = O (\log^{2} n)$
for every edge
$e \in E$.
\end{corollary}
\begin{proof}
For each of the $O(\log n)$ levels of the decomposition, the recursive SSSP
calls for each of the clusters can be merged into a single one by identifying
their super-sources. The claim hence follows from
Theorem~\ref{theorem:hierarchical-tsd} and
Corollary~\ref{cor:decomposition_congest}.
\end{proof}
From the hierarchical decomposition, we obtain embeddings into projected trees
and hierarchically $2$-separated trees as described
in Sections~\ref{section:projected-tree}~and~\ref{section:hst}.
\begin{corollary}\label{cor:projected_congest}
There exists a \Congest algorithm that, given a graph
$G = (V, E, \ell)$
with $\poly(n)$-bounded edge lengths, constructs a random projected
tree $T$ of $G$ in $\tilde{O}(\sqrt{n}+\hop(G))$ rounds that satisfies
the following guarantees for every edge $e \in E$:
(1)~$\Load_{T}(e) = O (\log n)$
w.h.p.;
and
(2)~$\Ex_{T} [\str_{T}(e)] = O (\log^{2} n)$.
\end{corollary}
\begin{proof}
We obtain an HTSD using Corollary~\ref{cor:htsd_congest}. Inspection of
the construction in Section~\ref{section:projected-tree} reveals that all
operations are local once we identify the leaders of clusters. This is, e.g.,
achieved by rooting all supporting trees at the respective cluster's leader,
which can be done by using the Garay-Kutten-Peleg minimum spanning tree
algorithm~\cite{garay1998sublinear,KuttenP98fast} to compute a spanning forest
of $H$. As the load of each edge is $O(\log n)$, the algorithm on $H$ can be
simulated at a multiplicative overhead of $O(\log n)$, resulting in running time
$\tilde{O}(\sqrt{n}+\hop(G))$.
\end{proof}
\begin{corollary}
There exists a \Congest algorithm that, given a graph
$G = (V, E, \ell)$
with $\poly(n)$-bounded edge lengths, constructs an embedding into a random
dominating hierarchically $2$-separated tree $T$ of $G$ in
$\tilde{O}(\sqrt{n}+\hop(G))$ rounds with expected stretch
$\Ex_{T} [\str_{T}(e)] = O(\log^{2} n)$ for each edge $e\in E$.
\end{corollary}
\begin{proof}
Analogous to Corollary~\ref{cor:projected_congest}.
\end{proof}

\subsection{\texorpdfstring{\Pram}{PRAM} Model}
In the \Pram model, multiple processors share a random access memory to jointly
solve a computational problem. Various contention models exist for concurrent
access to the same memory cell by multiple processors, but are equivalent up to
small (sub-logarithmic) factors in complexity, so we assume that there is no
contention. Then we can view the computation as a DAG whose nodes represent
elementary computational steps and edges dependencies. The input is represented
by the sources of the DAG. The crucial complexity measures are \emph{work,} the
total size of the DAG (or, equivalently, the sequential complexity of the
computation) and \emph{depth,} the maximum length of a path in the DAG (or,
equivalently, the time to complete the computation with an unbounded number of
processors executing steps at unit speed).

We use a result on approximate SSSP computations due to Cohen, who
introduced hop sets for this purpose. Following standard notation, we use
$m:=|E|$, where $G=(V,E,\ell)$ is the input graph.

\begin{corollary}[of~\cite{DBLP:journals/jacm/Cohen00, ElkinN16b}]
\label{cor:sssp_pram}
  Let
  $\eps_0 > 0$ be a constant and $\eps = \frac{1}{\polylog n}$.
  Then $(1+\eps)$-approximate SSSP in super-source graphs can be solved
  in $O(m^{1+\eps_0})$ work and $\polylog n$ time w.h.p.
\end{corollary}
We remark that the assumption that the graph $G$ is connected implies
$m^{\eps_0}\geq n^{\eps_0}$ and thus the term $m^{\eps_0}$ can absorb
$\polylog n$ factors.

Following the same route as for the \Congest model, we obtain a string of
corollaries. As coordination between processes is easier in the \Pram model, in
most cases the results are immediate.
\begin{corollary}\label{cor:blur_pram}
Suppose $\alpha= \frac{1}{\polylog n}$, $\rho= n^{O(1)}$, and $\eps_0$ is a
constant. Then Algorithm~\ref{alg:blur} can be executed in the \Pram model with
depth $\polylog n$ and work $O(m^{1+\eps_0})$ w.h.p.
\end{corollary}

\begin{corollary}\label{cor:decomposition_pram}
If $\Delta= n^{O(1)}$ and $\eps_0$ is a constant,
Algorithm~\ref{alg:decomposition} can be executed in the \Pram model with depth
$\polylog n$ and work $O(m^{1+\eps_0})$ w.h.p.
\end{corollary}
Combining this corollary with Theorem~\ref{theorem:hierarchical-tsd}, we obtain
the following result.
\begin{corollary}\label{cor:htsd_pram}
Fix any constant $\eps_0>0$.
There exists a \Pram algorithm of depth $\polylog n$ and work
$O(m^{1+\eps_0})$ that, for a graph $G = (V, E, \ell)$
with $\poly(n)$-bounded edge lengths, constructs a random HTSD
$\Seq{D}$ of $G$ with the following guarantees w.h.p.:
(1)
the depth of $\Seq{D}$ is
$O (\log n)$;
(2)
$\Seq{D}$ admits a (deterministic) geometrically decreasing diameter bounding
sequence
$\Seq{d}$
w.h.p.;
(3)
$\Load_{\Seq{D}}(e) = O (\log n)$
for every edge
$e \in E$
w.h.p.;
and
(4)
$\Ex_{\Seq{D}} [ \str_{\Seq{D}, \Seq{d}}(e) ] = O (\log^{2} n)$
for every edge
$e \in E$.
\end{corollary}
\begin{corollary}\label{cor:projected_pram}
Fix any constant $\eps_0>0$.
There exists a \Pram algorithm of depth $\polylog n$ and work
$O(m^{1+\eps_0})$ that, given a graph
$G = (V, E, \ell)$
with $\poly(n)$-bounded edge lengths, constructs a random projected
tree $T$ of $G$ that satisfies the following
guarantees for every edge $e \in E$:
(1)
$\Load_{T}(e) = O (\log n)$
w.h.p.;
and
(2)
$\Ex_{T} [\str_{T}(e)] = O (\log^{2} n)$.
\end{corollary}
\begin{proof}
Again, the main step after obtaining an HTSD is to identify cluster leaders.
This can be easily done by pointer jumping within the stated complexity bounds.
\end{proof}
\begin{corollary}
Fix any constant $\eps_0>0$. There exists a \Pram algorithm of depth $\polylog
n$ and work $O(m^{1+\eps_0})$ that, given a graph $G = (V, E, \ell)$
with $\poly(n)$-bounded edge lengths, constructs an embedding into a random
dominating hierarchically $2$-separated tree $T$ of $G$ with expected stretch
$\Ex_{T} [\str_{T}(e)] = O (\log^{2} n)$ for each edge $e\in E$.
\end{corollary}

\subsection{Semi-Streaming Model}
In the \emph{streaming model}~\cite{HenzingerRR98,McGregor14},
the input graph is given as a stream of edges without repetitions.
The performance of an algorithm is measured by the space it uses,
whereby space is organized in memory words of $O(\log n)$ bits.
In the \emph{multipass streaming model},
the input is presented to the algorithm in several such passes,
and the goal is to keep both the number of required passes
and the space consumption small.
For algorithms for graph problems, it is usual to assume arbitrary
arrival order of the edges.
The special case where the computational problem takes an $n$-vertex graph as
input and the amount of memory is $\tilde{O}(n)$ is also known as the
\emph{semi-streaming model}~\cite{FeigenbaumKMSZ2005}.
All our results in this subsection are for this setting.

\begin{corollary}[of \cite{DBLP:conf/wdag/BeckerKKL17}]
In the semi-streaming model, $(1+\varepsilon)$-approximate SSSP in super-source
graphs can be solved in $\polylog n$ passes w.h.p.\ for any
$\eps = \frac{1}{\polylog n}$.
\end{corollary}

All computational steps that are not SSSP computations can be either directly
executed in memory (because only graphs of size $\tilde{O}(n)$ are involved) or
easily performed by storing $\polylog n$ words for each node and streaming once
(e.g., finding cluster leaders). Thus, corollaries analogous to the \Congest and
\Pram models are immediate.
\begin{corollary}\label{cor:blur_streaming}
Suppose $\alpha= \frac{1}{\polylog n}$, $\rho= n^{O(1)}$. Then
Algorithm~\ref{alg:blur} can be executed in the semi-streaming model with
$\polylog n$ passes w.h.p.
\end{corollary}
\begin{corollary}\label{cor:decomposition_streaming}
If $\Delta= n^{O(1)}$, Algorithm~\ref{alg:decomposition} can be executed in the
semi-streaming model with $\polylog n$ passes w.h.p.
\end{corollary}
\begin{corollary}\label{cor:htsd_streaming}
There exists a semi-streaming algorithm that, given a graph $G = (V, E, \ell)$
with $\poly(n)$-bounded edge lengths, constructs a random HTSD
$\Seq{D}$ of $G$ with the following guarantees in $\polylog n$ passes w.h.p.:
(1)
the depth of $\Seq{D}$ is
$O (\log n)$;
(2)
$\Seq{D}$ admits a (deterministic) geometrically decreasing diameter bounding
sequence
$\Seq{d}$
w.h.p.;
(3)
$\Load_{\Seq{D}}(e) = O (\log n)$
for every edge
$e \in E$
w.h.p.;
and
(4)
$\Ex_{\Seq{D}} [ \str_{\Seq{D}, \Seq{d}}(e) ] = O (\log^{2} n)$
for every edge
$e \in E$.
\end{corollary}
\begin{corollary}\label{cor:projected_streaming}
There exists a semi-streaming algorithm that, given a graph
$G = (V, E, \ell)$
with $\poly(n)$-bounded edge lengths, in $\polylog n$ passes constructs a random
projected tree $T$ of $G$ that satisfies the following
guarantees for every edge $e \in E$:
(1)
$\Load_{T}(e) = O (\log n)$
w.h.p.;
and
(2)
$\Ex_{T} [\str_{T}(e)] = O (\log^{2} n)$.
\end{corollary}
\begin{corollary}
There exists a semi-streaming algorithm that, given a graph $G = (V, E, \ell)$
with $\poly(n)$-bounded edge lengths, in $\polylog n $ passes constructs an
embedding into a random dominating hierarchically $2$-separated tree $T$ of $G$
with expected stretch $\Ex_{T} [\str_{T}(e)] = O (\log^{2} n)$ for each edge
$e\in E$.
\end{corollary}

%%% Local Variables:
%%% mode: latex
%%% TeX-master: "embeddable_tree"
%%% End:

%!TEX root = ./embeddable_tree.tex

%%%%%%%%%%%%%%%%%%%%%%%%%%%%%%%%%%%%%%%
\section{Related Work}
\label{section:related-work}
%%%%%%%%%%%%%%%%%%%%%%%%%%%%%%%%%%%%%%%
Low diameter graph decompositions with small edge cutting probabilities (or
with small weight) play a major role in many algorithmic applications.
These include the construction of low stretch spanning trees
\cite{AbrahamBN2008, AbrahamN2012, AlonKPW1995, BeckerEGL2019, ElkinEST2008}
and low distortion probabilistic embeddings of metric spaces into
hierarchically well-separated trees
\cite{Bartal1996, Bartal1998, Bartal2004, FakcharoenpholRT2004},
fast approximate solvers of symmetric diagonally dominant linear systems
\cite{CohenKMPPRX2014, KoutisMP2011, KoutisMP2014, SpielmanT2004},
constructing graph spanners
\cite{MillerPVX2015,PelegS1989},
and spectral sparsification
\cite{KapralovP2012, Koutis2014}.
The literature in this field being vast, we can only give an
incomplete review of it.
We first focus on related work in distributed and parallel models of
computation, as these results are closest to ours, and then turn to the
related work in the streaming model.
Our discussion of the related work in the former models starts with reviewing
the literature on low diameter graph decompositions, and then it turns to
their applications, focusing on low average stretch spanning trees and tree
embeddings.

%###
\paragraph{Low Diameter Graph Decompositions.}
%###
In the \Local and \Congest models of distributed
computation\footnote{See
Section~\ref{subsec:congest} for the formal definition of \Congest. The \Local
model is identical, except that it does not restrict message sizes to
$O(\log n)$.}
low diameter graph decompositions for unweighted graphs, i.e., $G=(V,E,\ones)$,
play a special role
as they can be leveraged to design fast algorithms for a large class of
problems.
More precisely, the decomposition task is complete for a certain class of
local problems~\cite{GhaffariKM17}, where a problem is called local if it does
not require $\Omega(\hop(G))$ rounds of communication (recall the definition
of the hop diameter $\hop(G)$ from Section~\ref{subsec:congest}).
Here, $\hop(G)$ is of relevance even in problems where the input graph is
weighted, as communication over large hop distances is an
inherent obstacle to small running times in distributed algorithms.
% To maintain a clear distinction, we refer to the
% ``unweighted'' diameter of a graph as the \emph{hop diameter} $\hop(G)$.
%  Both models
% assume that each network node hosts its own processor, is given the ``local''
% part of the input (i.e., its neighbors, weights of incident edges, etc.) and
% needs to compute its local part of the output only (e.g., for our problem its
% parents in the supporting trees it participates in). Communication takes place
% over incident edges in synchronous rounds, implying that the number of rounds limits
% ``how far a node can see'' when taking its decision. In the \Local model, nodes
% base their output on full information up to this distance. In the \Congest
% model, message size is limited to $O(\log n)$ bits, requiring the algorithm to
% carefully select what is communicated.
% Decompositions of small diameter that cut
% few edges play an important role in such distributed models of computation
% as they can be leveraged to design fast algorithms for a large class of
% problems. More precisely, the decomposition task is complete
% for a certain class of local problems~\cite{GhaffariKM17}, where a
% problem is called local if it does not require $\Omega(\hop(G))$ rounds of
% communication, where $\hop(G)$ is the
% ``unweighted'' diameter of $G$, called the \emph{hop diameter}.

Several distributed decomposition algorithms with round complexities of
$\polylog n$ and small edge cut probabilities are known for the unweighted
case~\cite{ElkinN16,LinialS93,DBLP:conf/spaa/MillerPX13}.\footnote{Some works
only care about the chromatic number of the graph resulting from contracting
clusters. However, the cited works achieve this by cutting few edges only.}
However, the weighted setting considered in this work is fundamentally
different. A lower bound of $\hop(G)$ is trivial, i.e., the task is not local:
intuitively, decoupling hop distance from graph distance implies that finding
close-by nodes may require communication over $\hop(G)$ hops. In the \Local
model, this bound is trivially tight, as nodes can learn about the entire graph
in $\hop(G)$ rounds. In the \Congest model, a reduction from $2$-party
communication complexity shows a lower bound of
$\Omega\left(\frac{\sqrt{n}}{\log n}\right)$ rounds for computing an $(r,
\lambda)$-decomposition for any non-trivial values of $r$ and $\lambda$.
This lower bound even holds if $\hop(G)=O(\log
n)$~\cite{SarmaHKKNPPW2012}.\footnote{A low-diameter decomposition can be used
to determine whether or not there is a light $s$-$t$ cut in the family of lower
bound graphs from~\cite{SarmaHKKNPPW2012}; $s$ and $t$ end up in the same
cluster if and only if there is no light cut between them, as otherwise their
distance is large.}

% A number of works seek to circumvent or mitigate the complexity bottleneck of
% requiring exact SSSP computation for problems closely related to low-diameter
% decompositions, such as metric tree embeddings, low average stretch spanning
% tree construction, or solvers for symmetric diagonally dominant linear systems.
% Miller et
% al.~\cite{DBLP:conf/spaa/MillerPX13} do so in the \Pram model,
% % \footnote{For our
% % purposes, we consider a \Pram computation as a DAG of simple instructions, where
% % the complexity measures are work (total number of nodes of the DAG) and depth
% % (length of the longest path in the DAG).}
% for which they present a decomposition
% method that is efficient in unweighted graphs.

Miller et al.~\cite{DBLP:conf/spaa/MillerPX13} show how to compute low diameter
graph decomposition with small edge cutting probabilities in unweighted graphs
in the \Pram model. Their approach relies on exact SSSP computations. Given
the current discrepancy in the state of the art of exact and approximate SSSP in
the \Pram model, it thus cannot lead to satisfying bounds in the weighted
setting.

%###
\paragraph{Low Stretch Spanning Trees.}
%###
Nevertheless, there has been some work applying decompositions in the vein of
Miller et al.\ in order to obtain low average stretch\footnote{The ratio of
distance in the tree to edge length, averaged over all edges.}
spanning trees for weighted graphs.
A construction by Alon et al.~\cite{AlonKPW1995} reduces weighted graphs to
unweighted (multi)graphs.
As a result Blelloch et al.~\cite{blellochGKMPT14} were able to give an
efficient \Pram construction of low stretch spanning trees based on the
decomposition technique by Miller et al.
As shown by Blelloch et al., computation of such trees is of use for efficient
\Pram solvers for symmetric diagonally dominant linear systems.
A similar connection was exploited by Ghaffari et
al.~\cite{GhaffariKKLP2018}, who transferred the approach of Blelloch et al.\ to
the \Congest model, obtaining a low average stretch tree construction that they
leveraged for approximate maximum flow computations.
A downside of the aforementioned approaches is
that the construction by Alon et al.\ suffers from a poor average
stretch of $2^{\Theta(\sqrt{\log n \log \log n})}$, resulting in respective
overheads in work and depth resp.\ round complexity in the two models when
applying the computed trees in further computations.

For the \Congest model, Becker et al.~\cite{BeckerEGL2019} gave a
construction of low-average stretch spanning trees that combines the
decomposition technique of Miller et al.\ with the star decomposition technique
of Elkin et al.~\cite{ElkinEST2008}. This approach achieves $\polylog n$ average
stretch. Again, the complexity of their approach is essentially determined by
an exact SSSP computation. Thus, the resulting algorithm is round-optimal up to
polylogarithmic factors in the
unweighted case (i.e., the running time is $\hop(G)\polylog n$),
while essentially matching the round
complexity of exact SSSP in the weighted case. Exact SSSP computation in the
\Congest model is still not too well understood, with the best upper bound of
$\tilde{O}(\min\{\sqrt{n\hop(G)},\sqrt{n}\hop(G)^{1/4}+n^{3/5}+\hop(G)\})$~\cite{ForsterN2018}
still being polynomially far from the $\tilde{\Omega}(\sqrt{n}+\hop(G))$ lower bound.
\paragraph{Tree Embeddings.}
%###
We apply our decomposition technique in order to obtain a metric tree embedding,
following the same route as Bartal~\cite{Bartal1996}, obtaining the same
$O (\log^{2} n)$
bound on the expected stretch (note that the bound in \cite{Bartal1996} holds
for any edge lengths whereas in the current paper, we make the simplifying
assumption that the ratio of the maximum to minimum edge length is
$\poly(n)$).
Bartal later improved this bound to
$O(\log n \log \log n)$~\cite{Bartal1998} and
subsequently to asymptotically optimal $O(\log n)$~\cite{Bartal2004}.
Although we cannot readily apply the same techniques, Bartal's work suggests
that future improvements to our stretch bound are feasible.

Fakcharoenphol et al.~\cite{FakcharoenpholRT2004} achieved the $O(\log n)$
stretch bound earlier, following a different approach in which the graph is not
(explicitly) decomposed. However, at its core the main idea is very similar:
randomization is leveraged to keep the probability of ``cutting'' edges
proportional to their length based on the subtractive form of the triangle
inequality. Also here, \Pram and \Congest algorithms have been developed that
try to mitigate the bottleneck imposed by exact SSSP computations. In the
\Congest model, it is straightforward to implement the algorithm
from~\cite{FakcharoenpholRT2004} with a round complexity that is (up to a factor
of $O(\log n)$) equal to the running time of the Bellman-Ford
algorithm~\cite{KhanKMPT2012}. However, shortest paths may have hop length up to
$n-1$, resulting in a running time far from the
$\tilde{\Omega}(\sqrt{n}+\hop(G))$ lower bound.
Ghaffari and Lenzen broke down shortest paths by sampling a ``skeleton'' of
$\tilde{\Theta}(\sqrt{n})$ nodes uniformly, computing a spanner (refer to the
sequel of this section for the definition of a spanner) of a graph
representing the induced metric, computing a tree embedding of this spanner,
and finally extending this embedding to one of the original graph with
modified weights via a Bellman-Ford computation.
This can be seen as distorting the original distance metric such
that it becomes sufficiently simple to solve exact SSSP fast, resulting in a
round complexity of $\tilde{O}(n^{0.5+\varepsilon}+\hop(G))$ for stretch
$O(\varepsilon^{-1}\log n)$.
In particular, by setting
$\varepsilon = \frac{1}{\log n}$,
the stretch and running time bounds match our results.
However, Ghaffari and Lenzen do not guarantee bounded load.
We also note that their approach is inherently limited to stretch
$\Omega(\log^2 n)$
when requiring a running time bound within
$\polylog n$
of the lower bound, as both spanners with a near-linear number of edges and
metric tree embedding must incur
$\Omega(\log n)$
stretch each.

Friedrichs and Lenzen~\cite{FriedrichsL2018}
provide fast \Pram and \Congest algorithms for tree embeddings with stretch
$O(\log n)$. The main difference to~\cite{GhaffariL2014}
is the use of hop sets~\cite{DBLP:journals/jacm/Cohen00} to provide
``shortcuts'' for distance computation that are not present in the original
graph.
Again, distances are then distorted by metric embeddings such that exact
distance computation by a Bellman-Ford style computation becomes efficient. This
leads to a $2^{O(\sqrt{\log n})}(\sqrt{n}+\hop(G))$-round algorithm in \Congest
and a \Pram algorithm of depth $\polylog n$ and work $O(m^{1+\varepsilon})$ (for any
fixed constant $\varepsilon>0$), where $m$ is the number of edges and
$\Omega(m)$ a trivial lower bound on the work. While the stretch guarantee is
better than in our case, it should be noted that also here fundamental barriers
limit this technique: lower bounds on the size of hop sets due to Abboud et
al.~\cite{AbboudBP18} imply that any hop-set based approach must incur running
time resp.\ work overheads of $2^{\Omega(\sqrt{\log n})}$.
Although in the \Pram model we suffer the same work overhead by relying on
hop-sets for the currently best known approximate SSSP algorithms
\cite{DBLP:journals/jacm/Cohen00, ElkinN16b},
our result shows that one can trade the additional log-factor in stretch for a
logarithmic load bound that the method of Friedrichs and Lenzen cannot
guarantee.

%###
%%%%%%%%%%%%%%%%%%%%%%%%%%%%%%%%%%%%%%%
\paragraph{Streaming Algorithms.}
%%%%%%%%%%%%%%%%%%%%%%%%%%%%%%%%%%%%%%%
%###
To the best of our knowledge, constructions of low diameter decompositions
with small edge cutting probabilities have not been addressed so far in the
semi-streaming literature.
This is also true for constructions of low stretch spanning trees and other
types of dominating trees (including embeddable trees and HSTs studied in the
current paper).
A related graph theoretic object whose construction has been studied in the
context of streaming algorithms is \emph{spanners}.
Similarly to low (average) spanning trees, spanners also provide a sparse
distance preserving representation of the graph, only that they are not
required to be trees.
On the other hand, their notion of distance preservation is stronger in the
sense that it is required to hold in the worst case, rather than on average.
Specifically, a $\kappa$-spanner of graph
$G = (V, E, \ell)$
is a spanning subgraph of $G$ that guarantees a stretch bound of at most
$\kappa$ for every edge in $E$.
One is typically interested in constructing $\kappa$-spanners with a small
number of edges, where
$O (n^{1 + 2 / (\kappa + 1)})$
edges is the asymptotically tight bound.
Streaming constructions of sparse spanners exist only for
unweighted graphs
\cite{Baswana2008, Elkin2011,FeigenbaumKMSZ2005},
as there the distance computations are typically restricted to the sparse
subgraph maintained by the algorithm.
A related notion in unweighted graphs, which has also been studied in the
streaming literature \cite{ElkinZ2006}, is an
$(\alpha, \beta)$-spanner,
where the distance between vertices
$u, v \in V$
in the spanner is required to be at most
$\alpha \cdot \dist_{G}(u, v) + \beta$
for every
$u, v \in V$.

%%% Local Variables:
%%% mode: latex
%%% TeX-master: "embeddable_tree"
%%% End:

\bibliography{../references}

\end{document}